\documentclass[journal]{IEEEtran}

\usepackage{amsmath}
\usepackage{amssymb}
\usepackage{amsthm}
\usepackage{mathrsfs}
\usepackage{mathtools}
\usepackage{color}
\usepackage{wrapfig}
\usepackage{array}
\usepackage{scrextend}
\usepackage{multirow}
\usepackage{xcolor}
\usepackage{tikz,subcaption}
\usepackage{enumerate}
\usepackage{balance}
\usetikzlibrary{shapes,arrows,positioning,calc,patterns}

\tikzstyle{block} = [rectangle, draw, text centered, rounded corners, minimum height=2em,fill=black!5!]

\newtheorem{defn}{Definition}
\newtheorem{thm}{Theorem}
\newtheorem{lemma}{Lemma}
\newtheorem{corollary}{Corollary}
\newtheorem{construction}{Construction}
\newtheorem{example}{Example}
\newtheorem{remark}{Remark}

\newcommand{\A}{\ensuremath{\mathrm{A}}}
\newcommand{\C}{\ensuremath{\mathrm{C}}}
\newcommand{\G}{\ensuremath{\mathrm{G}}}
\newcommand{\T}{\ensuremath{\mathrm{T}}}

\newcommand{\ve}[1]{\ensuremath{\mathbf{#1}}}
\newcommand{\e}{\ensuremath{\mathrm{e}}}
\newcommand{\Sq}{\ensuremath{\Sigma_q}}

\newcommand{\ET}{\ensuremath{\mathbb{T}}}
\newcommand{\SE}{\ensuremath{\mathcal{E}}}
\newcommand{\SL}{\ensuremath{\mathcal{L}}}
\newcommand{\SG}{\ensuremath{\mathcal{C}}}

\newcommand{\Code}{\ensuremath{\mathscr{C}}}

\newcommand{\ins}{\ensuremath{\mathbb{I}}}
\newcommand{\del}{\ensuremath{\mathbb{D}}}
\newcommand{\sub}{\ensuremath{\mathbb{S}}}

\newcommand{\ff}[2]{\ensuremath{#1^{\underline{#2}}}}

\newcommand{\avg}[1]{\ensuremath{\mathbb{E}[#1]}}

\newcolumntype{C}[1]{>{\centering\let\newline\\\arraybackslash\hspace{0pt}}m{#1}}

\usepackage{etoolbox}

\usepackage{booktabs, multirow}

\begin{document}
\title{Coding over Sets for DNA Storage}

\author{Andreas~Lenz,~
        Paul~H.~Siegel,~
        Antonia Wachter-Zeh,~
        and~Eitan~Yaakobi~%
        \thanks{This paper was presented in part at the 2018 International Symposium on Information Theory \cite{LSWY18}, at the 2019 Information Theory and Applications Workshop, and the 2019 Non-Volatile Memories Workshop.}
\thanks{A. Lenz is with the Institute for Communications Engineering, Technische Universit\"at M\"unchen, Munich 80333, Germany (e-mail: andreas.lenz@mytum.de).}%
\thanks{P.  H.  Siegel  is  with  the  Electrical  and  Computer  Engineering  Department and the Center for Memory and Recording Research, University of California, San Diego, La Jolla, CA 92093-0407 USA (e-mail:	psiegel@ucsd.edu).}%
\thanks{A. Wachter-Zeh is with the Institute for Communications Engineering, Technische Universit\"at M\"unchen, Munich 80333, Germany  (e-mail: antonia.wachter-zeh@tum.de).}
\thanks{E. Yaakobi is with the Computer Science Department, Technion -- Israel Institute of	Technology,	Haifa 32000, Israel	(e-mail: yaakobi@cs.technion.ac.il).}%
\thanks{E. Yaakobi was supported by the Center for Memory and Recording Research, University of California San Diego. This project has received funding from the European Research Council (ERC) under the European Union’s Horizon 2020 research and innovation programme (grant agreement No 801434). This work was also supported by NSF Grant CCF-BSF-1619053 and by the United States-Israel BSF grant 2015816.}%
}

\maketitle

\begin{abstract}
	In this paper we study error-correcting codes for the storage of data in synthetic deoxyribonucleic acid (DNA). We investigate a storage model where a data set is represented by an unordered set of $M$ sequences, each of length~$L$. Errors within that model are a loss of whole sequences and point errors inside the sequences, such as insertions, deletions and substitutions. We derive Gilbert-Varshamov lower bounds and sphere packing upper bounds on achievable cardinalities of error-correcting codes within this storage model. We further propose explicit code constructions than can correct errors in such a storage system that can be encoded and decoded efficiently. Comparing the sizes of these codes to the upper bounds, we show that many of the constructions are close to optimal.
\end{abstract}\vspace{-.5ex}
\begin{IEEEkeywords}
coding over sets, DNA data storage, Gilbert-Varshamov bound, insertion and deletion errors, sphere packing bound
\end{IEEEkeywords}

\IEEEpeerreviewmaketitle

\section{Introduction}
\vspace{0ex}
DNA-based storage has attracted significant attention due to recent demonstrations of the viability of storing information in macromolecules. This recent increased interest was paved by substantial progress in synthesis and sequencing technologies.
The main advantages of DNA-based storage over classical storage technologies are very high data densities and long-term reliability without electrical supply. Given the trends in cost decreases of DNA synthesis and sequencing, it is now acknowledged that within the next $10$--$15$ years DNA storage may become a highly competitive archiving technology \cite{Carmean2019}.

A DNA storage system consists of three important entities (see Fig.~\ref{fig:DNAStorage}): (1) a DNA synthesizer that produces the strands that encode the data to be stored in DNA. In order to produce strands with acceptable error rate the length of the strands is typically limited to no more than 250 nucleotides (cf. \cite{Carmean2019} and also see Table \ref{tab:MLbeta} for an overview over current experiments); (2) a storage container with compartments that store the DNA strands, although in an unordered manner; (3) a DNA sequencer that reads the strands and transfers them back to digital data. The encoding and decoding stages are external processes to the storage system which convert the binary user data into strands of DNA in a way that even in the presence of errors, 
it is possible to reconstruct the original data. 

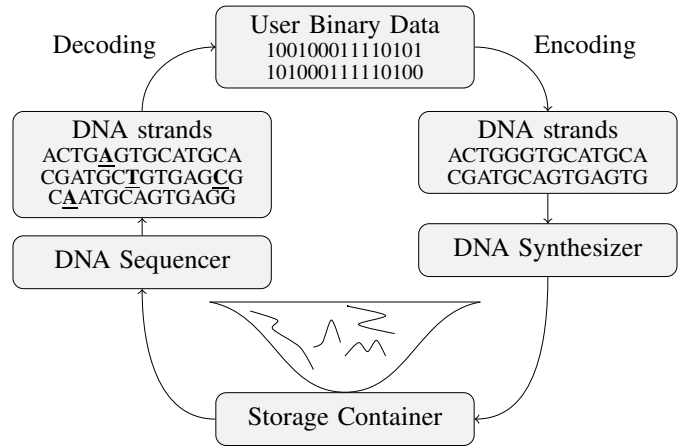
\begin{figure}%
	\centering
	\newcommand{\scale}{0.6}
\begin{tikzpicture}[node distance = 3cm, auto, text width = 3.2cm, line width=0.125mm]
\node [block, align=center]  (data) {{User Binary Data\\\footnotesize100100011110101\\[-1ex]\footnotesize101000111110100}};

\node [block, align=center, below of= data, node distance = 4.95cm] (con) {Storage Container};
\node [above= 0.0cm of con] (conpic) {};

\node [block, align=center, below right=0.3cm and -0.75cm of data, node distance = 2.5cm] (strands) {{DNA strands\\\footnotesize ACTGGGTGCATGCA \\[-1ex] \footnotesize CGATGCAGTGAGTG} };

\node [block, align=center, below of= strands, node distance = 1.3cm] (synth) {DNA Synthesizer};

\node [block, align=center, below left=0.3cm and -0.75cm of data, node distance = 2.5cm] (strands2) {{DNA strands\\
	\footnotesize ACTG\underline{\bf A}GTGCATGCA \\[-1ex]
	\footnotesize CGATGC\underline{\bf T}GTGAG\underline{\bf C}G}\\[-1ex]
	\footnotesize C\underline{\bf A}ATGCAGTGAGG };

\node [block, align=center, below of= strands2, node distance = 1.3cm] (seq) {DNA Sequencer};

\draw [->] (data.east) to [out=0,in=90] (strands.north) ;
\draw [->] (strands.south) to (synth.north);
\draw [->] (synth.south) to [out=270,in=0] (con.east);
\draw [->] (con.west) to [out=180,in=270] (seq.south);
\draw [->] (seq.north) to (strands2.south);	
\draw [->] (strands2.north) to [out=90,in=180] (data.west);

\node[] at (strands |- data) {\hspace{1.3cm} Encoding};
\node[] at (strands2 |- data) {\hspace{0.3cm} Decoding};
\draw ($(conpic.south) + (-3*\scale,2*\scale)$) to[out=350,in=180] (conpic.south) to[out=0,in=190] ($(conpic.south) + (3*\scale,2*\scale)$);
\draw ($(conpic.south) + (-3*\scale,2*\scale)$) --($(conpic.south) + (3*\scale,2*\scale)$);

\draw plot [smooth] coordinates {($(conpic.south) + (-0.7*\scale,1*\scale)$) ($(conpic.south) + (-.5*\scale,1.1*\scale)$) ($(conpic.south) + (-.3*\scale,1.6*\scale)$) ($(conpic.south) + (-.1*\scale,1.1*\scale)$)};

\draw plot [smooth] coordinates {($(conpic.south) + (0*\scale,0.8*\scale)$) ($(conpic.south) + (0.3*\scale,1.1*\scale)$) ($(conpic.south) + (0.5*\scale,0.9*\scale)$) ($(conpic.south) + (0.7*\scale,1.15*\scale)$) ($(conpic.south) + (0.9*\scale,0.8*\scale)$)};

\draw plot [smooth] coordinates {($(conpic.south) + (1.1*\scale,1.3*\scale)$)($(conpic.south) + (0.3*\scale,1.5*\scale)$) ($(conpic.south) + (0.9*\scale,1.7*\scale)$) ($(conpic.south) + (0.1*\scale,1.9*\scale)$) ($(conpic.south) + (0.1*\scale,1.9*\scale)$)};

\draw plot [smooth] coordinates {($(conpic.south) + (-2.1*\scale,1.8*\scale)$) ($(conpic.south) + (-1.5*\scale,1.6*\scale)$) ($(conpic.south) + (-1.6*\scale,1.4*\scale)$) ($(conpic.south) + (-1*\scale,1*\scale)$) ($(conpic.south) + (-0.7*\scale,0.5*\scale)$)};
\end{tikzpicture}
	\caption[DNA Storage system]
	{Illustration of a DNA-based storage system.}\vspace{-3.5ex}
	\label{fig:DNAStorage}%
\end{figure}

DNA as a storage system has several attributes which distinguish it from any other storage system. The most prominent one is that the strands are not ordered in the memory and thus it is not possible to know the order in which they were stored. One way to address this problem is using block addresses, also called indices, that are stored as part of the strand.
Errors in DNA are typically substitutions, insertions, and deletions, where most published studies report that either substitutions or deletions are the most prominent ones, depending upon the specific technology for synthesis and sequencing~\cite{BGHCTIPC16,EZ17,KC14,Oetal17,RRCHLHNJ13,YYLWLCKC14}. For example, in column-based DNA oligo synthesis the dominant errors are deletions that result from either failure to remove the dimethoxytrityl (DMT) or combined inefficiencies in the coupling and capping steps~\cite{KC14}. While codes correcting substitution errors were widely studied, much less is known for codes correcting deletions and insertions. The task of error correction becomes even more challenging taking into account the lack of ordering of the strands. 

{\bf Related work:} For a general survey about DNA-based data storage, the reader is referred to \cite{YKGMZM15}. The first large scale experiments that demonstrated the potential of \emph{in vitro} DNA storage were reported by Church et al. who recovered 643 KB of data~\cite{CGK12} and Goldman et al. who accomplished the same task for a 739 KB message~\cite{GBCDLSB13}. However both of these groups did not recover the entire message successfully due to the lack of using the appropriate coding solutions to correct errors. Church et al. had 10 bit errors and Goldman et al. lost two strands of 25 nucleotides.
Later, in~\cite{GHPPS15}, Grass et al. reported the first system with usage of error-correcting codes in DNA-based storage and managed to perfectly recover an 81 KB message. Bornholt et al. similarly retrieved a 42 KB message~\cite{BLCCSS16}. 
Since then, several groups have built similar systems, storing ever larger amounts of data.  Among these, Erlich and Zielinski~\cite{EZ17}  stored 2.11MB of data with high storage rate,  Blawat et al.~\cite{BGHCTIPC16} successfully stored 22MB, and more recently Organick et al.~\cite{Oetal17} stored 200MB.
Yazdi et al.~\cite{YTMZM15,YGM17} developed a method that offers both random access and rewritable storage. On the other hand, coding theoretic aspects of DNA storage systems have received significant attention recently. The work of \cite{KPM16} discusses error-correcting codes for the DNA sequencing channel, where a possibly erroneous collection of substrings of the original sequence is obtained. In \cite{KT18}, unordered multisets with errors that affect the whole sequence have been discussed. Furthermore, the model proposed in this work has already been adopted in~\cite{SRB18,SC2018}. Namely, codes and bounds for an arbitrary number of substitutions in sets of DNA strands have been derived in \cite{SRB18} and it has been shown that it is possible to construct codes, which have logarithmic redundancy on both, the number of sequences and the length of the sequences. In \cite{SC2018}, a distance measure for the DNA storage channel has been discussed and Singleton-like and Plotkin-like code size upper bounds have been derived. In contrast, the goal of this work is to study and to design error-correcting codes which are specifically targeted towards the special structure of DNA storage systems. This goal is accomplished by deriving upper and lower bounds on the achievable size of error-correcting codes and designing constructions over sets that are suitable for data storage in DNA. Errors within this model are a loss of sequences and point errors inside the sequences, such as insertions, deletions, and substitutions. Parts of this work have been published in \cite{LSWY18}, at the 2019 Information Theory and Applications Workshop, and at the 2019 Non-Volatile Memories Workshop.

The paper is organized as follows. We start by introducing the DNA storage channel model and associated notation. In Sections \ref{sec:gv:bounds} and \ref{sphere:packing:bounds} we derive generalized Gilbert-Varshamov bounds and sphere packing bounds for the DNA storage channel, which bound the cardinality of optimal error-correcting codes, i.e., codes of largest possible cardinality from below and above. Then, in Section \ref{sec:const}, we propose code constructions that can correct errors from the DNA storage channel. Lastly, Section~\ref{sec:concl} concludes the paper.
\section{Channel Model} \label{sec:channel:model}
\subsection{Notation}
We start by introducing the notation that will be used throughout the paper. For any sets $\mathcal{A}, \mathcal{B}$ we write $|\mathcal{A}|$ as the cardinality of $\mathcal{A}$ and $\mathcal{A}\setminus \mathcal{B} = \{ x: x\in \mathcal{A} \land x \notin \mathcal{B} \}$ as the set difference. We denote by $\mathbb{N}$ and $\mathbb{N}_0$ the sets of natural numbers, where the former consists of the numbers $\{1,2,3,\dots\}$ and the latter additionally contains $0$. The set $[n] = \{1,2,\dots,n\}$ contains all natural numbers up to $n \in \mathbb{N}$. $\Sq$ is a finite alphabet with $q$ elements. In particular, we will write $\Sigma_2= \{0,1\}$ for binary sequences and $\Sigma_4 = \{\A,\C,\G,\T\}$ for DNA sequences. A vector of $n$ elements $x_i \in \Sq$ over an alphabet $\Sq$ is denoted by $\ve{x} = (x_1,x_2, \dots, x_n) \in \Sq^n$. Its first, respectively last $m$ elements are denoted by $\mathrm{pref}_m(\ve{x})$ and $\mathrm{suff}_m(\ve{x})$. The number of runs in $\ve{x} \in \Sq^n$, is denoted as $||\ve{x}|| \triangleq |\{i : x_{i}\neq x_{i+1}\}|+1$. For two vectors $\ve{x} \in \Sq^n, \ve{y} \in \Sq^m$ we write $(\ve{x},\ve{y})$ as the concatenation of $\ve{x}$ and $\ve{y}$ which has length $n+m$. Throughout the paper, we denote the binary logarithm of a real number $a\in \mathbb{R}^+$ by $\log(a)$ and the natural logarithm by $\ln(a)$. For any integers $n,m \in \mathbb{N}$, $m\leq n$ we write $n!=n\cdot(n-1)\dots2 \cdot 1$ as the factorial and $\ff{n}{m} = n(n-1)\dots(n-m+1)$ as the falling factorial. The binomial coefficient is denoted by $\binom{n}{m}$. For the asymptotic behavior of functions, we use the Bachmann-Landau notation, i.e., for $f(n), g(n) : \mathbb{N} \mapsto \mathbb{R}$, we write
\begin{itemize}
	\item $f(n) = o(g(n))$, if $\lim\limits_{n\rightarrow \infty} \frac{f(n)}{g(n)} = 0$,
	\item $f(n) = \omega(g(n))$, if $\lim\limits_{n\rightarrow \infty} \left|\frac{f(n)}{g(n)}\right| = \infty$,
	\item $f(n) = O(g(n))$, if $\limsup\limits_{n\rightarrow \infty} \left|\frac{f(n)}{g(n)}\right| < \infty$,
	\item $f(n) \sim g(n)$, if  $\lim\limits_{n\rightarrow \infty} \frac{f(n)}{g(n)} = 1$, and
	\item $f(n) \gtrsim g(n)$, if  $\lim\limits_{n\rightarrow \infty} \frac{f(n)}{g(n)} \geq 1$.
\end{itemize}%
\subsection{DNA Channel Model}
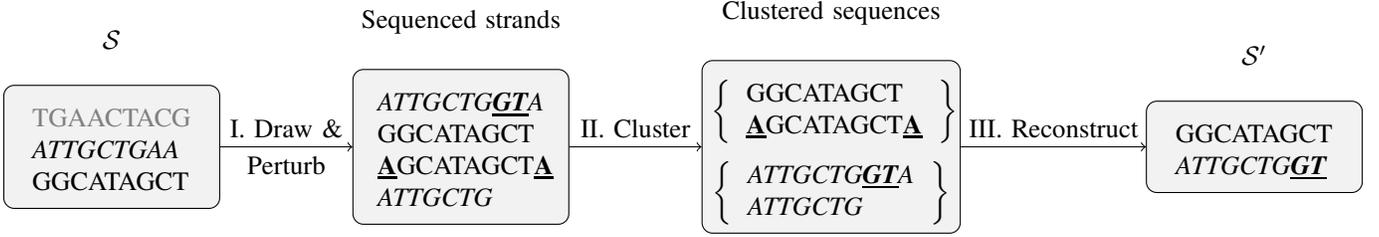
\begin{figure*}[]
	
	\centering
	\begin{tikzpicture}[node distance=1em,auto]

	\node[block, text width =7.5em] (S) { \vspace{-.1cm}
		$$\begin{tabular}{l}
		\textcolor{gray}{TGAACTACG} \\ \textit{ATTGCTGAA} \\ {GGCATAGCT}
		\end{tabular}$$%
	};
	\node[above= of S] (lS) {$\mathcal{S}$};
	
	\node[block,right= 5em of S, text width =7.5em](R) { \vspace{-.1cm}
		$$\begin{tabular}{l}
		\textit{ATTGCTG\underline{\textbf{GT}}A} \\ {GGCATAGCT} \\ {\underline{\textbf{A}}GCATAGCT\underline{\textbf{A}}} \\ \textit{ATTGCTG}
		\end{tabular}$$%
	};
	\node[above= of R] (lR) {Sequenced strands};
	
	\node[block,right= 5em of R, text width =9.1em](C) { \vspace{-.1cm}
		$$\left\{\begin{tabular}{l}
		{GGCATAGCT} \\ {\underline{\textbf{A}}GCATAGCT\underline{\textbf{A}}}
		\end{tabular}\right\}$$%
		$$\left\{\begin{tabular}{l}
		\textit{ATTGCTG\underline{\textbf{GT}}A} \\ \textit{ATTGCTG}
		\end{tabular}\right\}$$%
	};
	\node[above= of C] (lC) {Clustered sequences};
	
	\node[block,right= 7em of C, text width =7.5em](Sp) {\vspace{-.1cm}
		$$\begin{tabular}{l}
		{GGCATAGCT} \\ \textit{ATTGCTG\underline{\textbf{GT}}}
		\end{tabular}$$%
	};
	\node[above= of Sp] (lSp) {$\mathcal{S}'$};
	
	\draw [->] (S) -- (R) node[midway,above] {I. Draw \&} node[midway, below] {Perturb};
	\draw [->] (R) -- (C) node[midway,above] {II. Cluster};
	\draw [->] (C) -- (Sp) node[midway,above] {III. Reconstruct};
	
	\end{tikzpicture}
	
	\caption{DNA storage channel model. Sequences with the same text decoration stem from the same original sequence.}
	\label{fig:channel:model}
	
\end{figure*}
We consider the DNA storage channel, which is depicted in Fig.~\ref{fig:channel:model}. In a DNA-based data storage system, data is stored in an unordered \emph{set}
\begin{equation*}
\mathcal{S} = \{ \mathbf{x}_1, \mathbf{x}_2, \dots, \mathbf{x}_M \} \subseteq \Sigma_q^L,
\end{equation*}
with $M$ distinct \textit{sequences} $\mathbf{x}_i \in \Sigma_q^L$, i.e. $\ve{x}_i \neq \ve{x}_j$ for $i\neq j$. Each sequence $\ve{x}_i$ has length $L$. Here and in the rest of the paper whenever we write the set $\mathcal{S}$ we assume it is a set of $M$ sequences as defined above. Throughout the paper, we will refer to the $\ve{x}_i$ by \textit{sequences} or \textit{strands} and to $\mathcal{S}$ by data \textit{sets} or \textit{words}. Representing data words as unordered sets is inherently natural, due to the following two reasons. First, any information about ordering of the data sequences is lost during the storage and second, in the reading process it is not possible to distinguish exactly how many times each sequence was stored, since the sequences are multiplied in the storage medium and not necessarily all of them are read. For more details on the channel model, see~\cite{HRT17,HMG18}. 

Any such stored data set $\mathcal{S}$ of $M$ sequences is a possible input of the DNA storage channel. Hence, the input space, which comprises all possible data sets is denoted by
\begin{equation*}
\mathcal{X}_M^L = \{\mathcal{S} \subseteq \Sigma_q^L : |\mathcal{S}| = M\}.
\end{equation*}
The DNA storage channel can be split into the three following stages, as visualized in Fig. \ref{fig:channel:model}.
\begin{enumerate}[I.]
	\item Random sequences are drawn with replacement from the storage medium $\mathcal{S}$ and sequenced, possibly with substitution, insertion or deletion errors.
	\item The sequenced strands are clustered according to their Levenshtein distance. \footnote{This technique was introduced in \cite{Oetal17}, exploiting the fact that sequences are drawn several times. Other works have either clustered the sequences according to their indices (as in \cite{GHPPS15}), directly used a code on each sequence to correct insertions and deletions (as in \cite{YGM17}), or simply discarded sequences of incorrect length.}
	\item The clustered sequences are reconstructed by performing an estimate $\ve{x}'$ for each cluster, resulting in the received estimates $\mathcal{S}'$. If two or more reconstructions result in the same estimate $\ve{x}'$, we only output a single sequence $\ve{x}'$ to avoid possible duplicates of a single stored sequence. Therefore, $\mathcal{S}'$ is a set with distinct elements.
\end{enumerate}
In this work we consider the combination of the above three stages, from the stored sequences $\mathcal{S}$ to the reconstructed sequences $\mathcal{S}'$, as the DNA storage channel. Each sequence $\ve{x} \in \mathcal{S}$ is therefore either
\begin{itemize}
	\item reconstructed correctly, without errors ($\ve{x} \in \SG$),
	\item never drawn or its cluster is not identified and thus lost in the storage medium ($\ve{x} \in \SL$), or
	\item reconstructed with errors ($\ve{x} \in \SE$),
\end{itemize}
where $(\SG,\SL, \SE)$ is a partition of $\mathcal{S}$.

According to the above three cases, we thus associate the following three parameters $(s,t,\epsilon)_\ET$ that characterize the DNA storage channel. We denote by $s$ the maximum number of sequences that are never drawn (or whose clusters are not identified), by $t$ the maximum number of sequences that have been reconstructed with errors, and by $\epsilon$ the maximum number of errors of type $\ET$ in each of the latter. Typical error types $\ET$ after the reconstruction step are various combinations of insertions, deletions and substitutions, where the latter two are the most prominent ones in DNA storage systems \cite{Oetal17}.
To be more precise, we define the error balls associated with the channel model. We start with the characterization of point errors inside the sequences.
\begin{defn} \label{def:error:ball}
	The \textit{error ball} $B^\ET_{\epsilon}(\ve{x})$ of radius $\epsilon$ around a sequence $\ve{x} \in \Sigma_q^L$ is defined to be the set of all possible outcomes $\ve{x}' \in B^\ET_{\epsilon}(\ve{x})$, after $\epsilon$ (or fewer) errors of type $\ET$ in $\ve{x}$. Possible types of errors are
	\begin{itemize}
		\item Insertions ($\ins$),
		\item Deletions ($\del$),
		\item Substitutions ($\sub$),
	\end{itemize}
	or combinations of the above, denoted by, e.g., $\ins\del$ for the case of insertions and deletions. We use the abbreviation $\mathbb{L} \triangleq \ins\del\sub$ for insertions, deletions, and substitutions. Similarly, we define the error sphere $S^\ET_{\epsilon}(\ve{x})$ as the set of possible results $\ve{x}' \in S^\ET_{\epsilon}(\ve{x})$ after exactly $\epsilon$ errors of type $\ET$. For uniform error balls and spheres, where the size does not depend on the center $\ve{x} \in \Sigma_q^L$ we use the abbreviation $B^\ET_{\epsilon}(L) \triangleq |B^\ET_{\epsilon}(\ve{x})|$, respectively $S^\ET_{\epsilon}(L) \triangleq |S^\ET_{\epsilon}(\ve{x})|$. In particular we have
	\begin{itemize}
		\item $S^\ins_{\epsilon}(L) = \sum_{i=0}^{\epsilon} \binom{L+\epsilon}{i}(q-1)^i$ (c.f. \cite{Lev74}),
		\item $B^\ins_{\epsilon}(L) = \sum_{i=0}^{\epsilon} S^\ins_{i}(L)$,
		\item $S^\sub_{\epsilon}(L) = \binom{L}{\epsilon}(q-1)^\epsilon$,
		\item $B^\sub_{\epsilon}(L) = \sum_{i=0}^{\epsilon} \binom{L}{i}(q-1)^i$.
	\end{itemize}
\end{defn}
Note that for the case of deletions, such an abbreviation is not possible, since the size of the deletion ball and sphere depends on the center $\ve{x}$. The following example illustrates the definitions of error balls for different error types.
\begin{example}
	Consider the  sequence $\ve{x} = (\A\C) \in \Sigma_4^2$ of length $L=2$ and a single error, $\epsilon=1$. The substitution error ball is given by $B_1^{\sub}(\ve{x}) = \{(\A\C), (\C\C), (\G\C), (\T\C), (\A\A), (\A\G), (\A\T)\}$. Similarly, the deletion ball around $\ve{x}$ is given by $ B_1^{\del}(\ve{x}) = \{(\A\C), (\C), (\A)\}.$ The insertion sphere around the center $\ve{x}$ is $S_1^\ins(\ve{x}) = \{(\A\A\C), (\C\A\C), (\G\A\C),(\T\A\C),(\A\C\C), (\A\G\C),(\A\T\C)$, $(\A\C\A),(\A\C\G),(\A\C\T)\}$.
\end{example}
In a similar fashion it is possible to define the error ball of a data set, as the set of possible received sets after the DNA storage channel.
\begin{defn} {\label{def:dna:channel:error:ball}}
	For $\mathcal{S}\in\mathcal{X}_M^L$, the \textit{error ball} $B^\ET_{s,t,\epsilon}(\mathcal{S})$ is defined to be the set of all possible received sets $\mathcal{S}'$ after $s$ (or fewer) sequences have been lost and $t$ (or fewer) sequences of the remaining sequences have been distorted by $\epsilon$ (or fewer) errors of type $\ET \in \{\sub, \ins, \del, \mathbb{ID}, \mathbb{IS}, \mathbb{DS}, \mathbb{L}\}$ each. 
	
	More precisely, let $\mathsf{Part}_{s,t}(\mathcal{S})$ be the set of all partitions $(\SG,\SL,\SE)$ of $\mathcal{S}$ with $|\SL| \leq s$, $|\SE| \leq t$ and denote by $\SE = \{\ve{x}_{e_1}, \ve{x}_{e_2},\dots, \}$ the set of stored sequences, which are received in error. We then define $B^\ET_{s,t,\epsilon}(\mathcal{S})$ to be
	\begin{align*}
	B^\ET_{s,t,\epsilon}(\mathcal{S}) \hspace{-.04cm}=\hspace{-.04cm} \left\{ \hspace{-.08cm}\mathcal{S}' = \SG \cup \SE' \left| \hspace{-.16cm} \begin{array}{l} (\SG,\SL,\SE) \in \mathsf{Part}_{s,t}(\mathcal{S}), \\
	\SE' = \{ \ve{x}_{1}'\} \cup \dots \cup \{\ve{x}_{{|\SE|}}' \} , \\ 
	\ve{x}_{i}' \in B^{\ET}_\epsilon(\ve{x}_{e_i}) \hspace{-.08cm}\setminus\hspace{-.08cm} \{\ve{x}_{e_i}\} ~ \forall ~ i \in [|\SE|]
	\end{array} \hspace{-.1cm} \right. \right\} \hspace{-.1cm}.	\end{align*}
	Hereby $\SE'$ satisfies $|\SE'|\leq |\SE|$ and denotes the set of all distinct erroneous received sequences $\ve{x}_i'$, after removing duplicates.
\end{defn}%
The erroneous sequences $\ve{x}_i'$ are not necessarily distinct from each other or from the correct sequences in $\SG$ and therefore it is possible that two erroneous sequences or one error-free and one erroneous sequence agree with one another, resulting in a loss of a sequence. The number of distinct received sequences $|\mathcal{S}'|$ therefore satisfies $M-t-s \leq |\mathcal{S}'| \leq M$.
\begin{example}
	Consider the example in Fig. \ref{fig:channel:model} for the DNA storage channel with $M=3$ stored sequences, $\ve{x}_1 = (\mathrm{TGAACTACG})$, $\ve{x}_2 = (\mathrm{ATTGCTGAA})$, and $\ve{x}_3 = (\mathrm{GGCATAGCT})$, each of length $L=9$, i.e., $\mathcal{S} = \{\ve{x}_1, \ve{x}_2, \ve{x}_3\} \in \mathcal{X}_3^9$. The sequenced strands are clustered and reconstructed, resulting in two estimates $\ve{y}_1 = (\mathrm{GGCATAGCT})$ and $\ve{y}_2 = (\mathrm{ATTGCTGGT})$. The received set is therefore $\mathcal{S}' = \{\ve{y}_1, \ve{y}_2\}$. Hereby $\ve{x}_3$ was received correctly as $\ve{y}_1$, $\ve{x}_1$ was lost and $\ve{x}_2$ was received in error as $\ve{y}_2$. It follows that the set of correct, lost and erroneous sequences is given by
	\begin{align*}
		\SG = \{\ve{x}_3\} &= \{(\mathrm{GGCATAGCT})\}, \\
		\SL = \{\ve{x}_1\} &=\{ (\mathrm{TGAACTACG}) \},\\
		\SE = \{\ve{x}_2\} &= \{(\mathrm{ATTGCTGAA})\}.
	\end{align*}
	It follows that $s=|\SL|=1$ and $t=|\SE|=1$, where there were $\epsilon=2$ substitution errors in $\ve{x}_2$. Therefore, $\mathcal{S}' \in B_{1,1,2}^\sub(\mathcal{S})$.
\end{example}
The channel from a stored set $\mathcal{S}$ to a received set $\mathcal{S}'$ is visualized in Fig. \ref{fig:simplified:channel:model}.
\begin{figure}
	
	\centering
	\begin{tikzpicture}[node distance=1em,auto]
	
	\node[block, text width=.5cm] (S) {$\mathcal{S}$};
	
	\node[circle,fill, inner sep=1.5pt,right= 4.5em of S] (circ) {};
	
	\node[block, text width=.5cm, above right= 2em and 3em of circ] (U) {$\SG$};
	\node[block, text width=.5cm, right= 2.87em of circ] (L) {$\SL$};
	\node[block, text width=.5cm, below right= 2em and 3em of circ] (F) {$\SE$};
	
	\node[block, text width=.5cm, right= 6em of F] (Fp) {$\SE'$};
	
	\node[block, text width=.5cm, right= 8em of L] (unify) {$\bigcup$};
	
	\node[block, text width=.5cm, right= 1em of unify] (Sp) {$\mathcal{S}'$};
	
	\draw [-] (S) -- (circ) node[midway,above] {Partition};
	
	\draw [->] (circ) |- (U);
	\draw [->] (circ) -- (L)  node[pos=.5,above] {$\leq s$};
	\draw [->] (circ) |- (F) node[pos=.76,above] {$\leq t$};
	
	\draw [->] (F) -- (Fp) node[midway, above] {Add $\leq \epsilon$} node[midway, below] {errors each};
	
	\draw [->] (U) -| (unify);
	\draw [->] (Fp) -| (unify);
	
	\draw [->] (unify) -- (Sp);

	\end{tikzpicture}
	\caption{Illustration of the $(s,t,\epsilon)_\ET$ channel model}
	\label{fig:simplified:channel:model}
	
\end{figure}
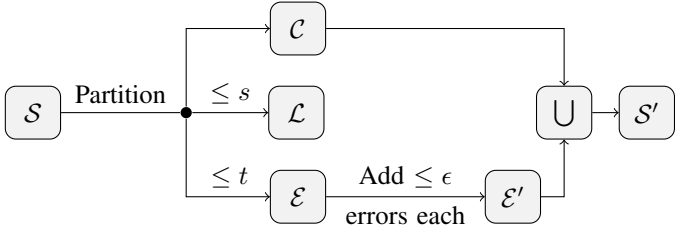
Throughout the paper, we will refer to the following definition of an error-correcting code in DNA storage systems.
\begin{defn} \label{def:correcting:code}
	A code $\Code \subseteq \mathcal{X}_M^L$ is called an \textbf{$(s,t,\epsilon)_\ET$-correcting code}, if it can correct a loss of $s$ (or fewer) sequences and $\epsilon$ (or fewer) errors of type $\ET$ in each of $t$ (or fewer) sequences, i.e., for any pair $\mathcal{S}_1, \mathcal{S}_2 \in \Code$ with $\mathcal{S}_1 \neq \mathcal{S}_2$, it holds that
	$$ B^\ET_{s,t,\epsilon}(\mathcal{S}_1) \cap B^\ET_{s,t,\epsilon}(\mathcal{S}_2) = \emptyset. $$
	We say $\Code \subseteq \mathcal{X}_M^L$ is an $(s,t,\bullet)_\ET$-correcting code if the number of errors $\epsilon$ per erroneous sequences can be arbitrarily large.
\end{defn}
Note that by this definition, a \textit{code} is a set of \textit{codewords}, where each \textit{codeword} is again a set of $M$ \textit{sequences} of length $L$. One of the main challenges associated with errors in such codewords is the loss of ordering information about the code sequences. The redundancy of a code is defined as follows.
\begin{defn} \label{def:redundancy}
	The \emph{redundancy} of a code $\Code \subseteq \mathcal{X}_M^L$ is \vspace{-0.5ex}
	$$ r(\Code) = \log|\mathcal{X}_M^L| - \log |\Code| =  \log\binom{q^L}{M} - \log |\Code|.\vspace{-0.5ex}$$
\end{defn}
We present the results in this work for binary sequences ($q=2$), however most or all of them can be extended to the non-binary case (and, in particular, the quaternary case). Our results about the redundancy of the proposed constructions and lower bounds on the redundancy are summarized in Table \ref{tab:construction}.
\subsection{Discussion of the Channel Model}

Designing and analyzing codes over sets allows to efficiently combat several important aspects of DNA-based data storage. These include the loss of the ordering information of the sequences and the loss or erroneous reception
of some of the stored sequences as described in our channel model. Especially when not all sequences are received with errors (i.e. some sequences are received correctly), it is not obvious at all, whether, e.g., prepending an index to each sequence is optimal and how the stored sequences should be protected from errors. Therefore, discussing the channel model from stored sets to received sets is of relevance when aiming for efficient and error-free data storage in DNA. Such a discussion is not possible when only the channel from a single stored sequence to a single received sequence is analyzed. 

The following remarks summarize two further observations about the channel model.
\begin{remark} \label{rem:order:L:M}
	While in practical DNA-based storage systems, the length of the sequences $L$ is moderate, e.g., in the order of a few hundreds, $M$ is significantly larger. In general, we say that $M = q^{\beta L}$ for some $0<\beta<1$. Typical values for the parameters $M,L$ and $\beta$ can be found in Table \ref{tab:MLbeta}.
\end{remark}
\begin{remark} \label{rem:channel:model}
	In view of the underlying DNA storage system, which is visualized in Fig. \ref{fig:channel:model}, the parameters $s,t,\epsilon$ of the channel model depend on the number of sequences that are drawn from the storage medium and also the reconstruction algorithm. Using an efficient reconstruction algorithm, it can be assumed that $s,t,\epsilon$ decrease as the number of draws increases, since the reconstruction can be performed more accurately. In particular, when many more than $M$ sequences are drawn from the storage medium, it can be assumed that there are enough draws per sequence that the sequencing errors are corrected by the reconstruction algorithm. Consequently there only remain errors which have been introduced when synthesizing the sequences.
\end{remark}
\subsection{Relationship of Insertion- and Deletion-Correcting Codes} \label{app:eq:ins:del}
In this section, we investigate the relationship between $(s,t,\epsilon)_\mathbb{I}$-insertion-correcting and $(s,t,\epsilon)_\mathbb{D}$-deletion-correcting codes. It is known \cite{Lev65} that for the case of standard {\nobreakdash blockcodes}, any code can correct $\epsilon$ insertions if and only if it can correct any $\epsilon$ insertions \emph{and} deletions. Surprisingly such an equivalence does \emph{not} hold for our channel model. Here we show a counterexample that an $(s,t,\epsilon)_\del$-correcting code is not necessarily an $(s,t,\epsilon)_\ins$-correcting code.
\begin{example}
	Consider the code $\Code = \{\mathcal{S}_1 , \mathcal{S}_2\}$, with $\mathcal{S}_1 = \{(\A\A\C\C\A),(\A\A\C\A\A),(\G\G\T\T\G)\}$ and $\mathcal{S}_2 = \{(\A\C\C\A\A),(\G\G\T\G\G),(\G\T\T\G\G)\}$. We can verify that $\Code$ is $(0,3,1)_\del$-correcting. It is however not $(0,3,1)_\ins$-correcting, since $\{(\A\A\C\C\A\A),(\G\G\T\T\G\G)\} \in B_{0,3,1}^{\ins}(\mathcal{S}_1)$ by editing both $(\A\A\C\C\A)$ and $(\A\A\C\A\A)$ to become $(\A\A\C\C\A\A)$ and $(\G\G\T\T\G)$ to become $(\G\G\T\T\G\G)$. Similarly, $\{(\A\A\C\C\A\A),(\G\G\T\T\G\G)\} \in B_{0,3,1}^{\ins}(\mathcal{S}_2)$, since we can edit $(\A\C\C\A\A)$ to become $(\A\A\C\C\A\A)$ and both $(\G\G\T\G\G)$ and $(\G\T\T\G\G)$ to become $(\G\G\T\T\G\G)$.
\end{example}
A counterexample for the other direction, i.e., an $(s,t,\epsilon)_\ins$-correcting code is not necessarily an $(s,t,\epsilon)_\del$-correcting code, can be found analogously.
\begin{table}
	\setlength{\tabcolsep}{4.7pt}
	\centering
	\caption{Summary of the parameters used in recent DNA storage experiments. The strand length is depicted as \emph{code} length which matches the length $L$ in our channel model.}
	{\renewcommand{\arraystretch}{2.1}
		\begin{tabular}{ccccc} \specialrule{.8pt}{0pt}{0pt}
			Work & Data Size & Strand Length $L$ & Strands $M$ & $\beta = (\log_4 M)/L$ \\ \specialrule{.8pt}{0pt}{0pt}
			\cite{CGK12} & $0.65$MB &  $115$ & $54,898$ & $0.0685$ \\
			\cite{GBCDLSB13} & $0.63$MB & $117$ & $153,335$ & $0.0736$ \\
			\cite{GHPPS15} & $0.08$MB & $117$ & $4,991$ & $0.0525$ \\
			\cite{YTMZM15} & $0.017$MB & $1000$ & $32$ & $0.0025$  \\
			\cite{BLCCSS16} & $0.15$MB & $120$ & $45,652$ & $0.0645$ \\
			\cite{BGHCTIPC16} & $22$MB & $190$ & $900,000$ & $0.0521$ \\
			\cite{EZ17} & $2.11$MB & $152$ & $72,000$ & $0.0531$ \\
			\cite{YGM17} & $0.003$MB & $1000$& $17$ & $0.0020$  \\
			\cite{Oetal17} & $200.2$MB & ${150 - 154}$ & $13,448,372$ & $0.0769-0.0789$ \\ \specialrule{.8pt}{0pt}{0pt}
	\end{tabular}}
	\label{tab:MLbeta}
\end{table}
\begin{table*}[t]
	\centering
	\caption{Lower and upper bounds on the redundancy of optimal $(s,t,\epsilon)_\ET$-correcting codes. Low order terms are omitted.}
	\setlength{\tabcolsep}{4pt}
	{\renewcommand{\arraystretch}{2.1}
		\begin{tabular}{l@{\hskip 5pt}c@{\hskip 5pt}lc@{\hskip 5pt}lcl}
			\specialrule{.8pt}{0pt}{0pt}
			Error correction & Gilbert-Varshamov bound & [Sect. \ref{sec:gv:bounds}] & Construction & [Sect. \ref{sec:const}] & Sphere packing bound & [Sect. \ref{sphere:packing:bounds}] \\ \specialrule{.7pt}{0pt}{0pt}
			\multirow{4}{*}{$(s,t,\bullet)_\mathbb{L}$}& \multirow{4}{*}{$(s+2t)L + (s+2t)\log M$} & \multirow{4}{*}{[Thm. \ref{thm:gv:s:t:bullet}]} & $M \log \e + (s+2t) (L-\lceil \log M \rceil)$ & [Const. \ref{con:index:rs}] & \multirow{4}{*}{$(s+t)L+t\log M$} & \multirow{4}{*}{[Cor. \ref{cor:bound:red:arbitrary:error}]} \\
			&&& $(s + 2 t)L$ & [Const. \ref{con:constant:weight}] &&  \\
			&&&   $\frac{(1-c)}{2}M^c\log M$ & \multirow{2}{*}{[Const. \ref{con:reduce:index}]} &&  \\
			&&& $+(s+2t) M^{1-c} \left( L-\log M\right)$ &&& \\ \hline
			$(\sigma M,\tau M, \bullet)_\mathbb{L}$& $(\sigma+2\tau)(L-\log M)$ & [Thm. \ref{thm:gv:s:t:bullet}] & $(\sigma+2\tau)M(L-\log M)$ & [Const. \ref{con:constant:weight}] & $(\sigma+\tau)M(L-\log M)$& [Cor. \ref{cor:bound:red:arbitrary:error}] \\ \hline
			$(s,t,\epsilon)_\sub$ & $sL+(s+2t)\log M +2t\epsilon \log L $ & [Thm. \ref{thm:gv:s:t:e:sub}] &  & & $sL+t\log M + t\epsilon \log L$ & [Thm. \ref{thm:substitutions:asymptotic}] \\ 
			$(s,t,\epsilon)_\del$ & $sL + (s+t) \log M +2t\epsilon \log (L/2) $ & [Thm. \ref{thm:gv:s:t:e:del}] & $(s+t)L$ & [Const. \ref{con:constant:weight}] & $sL + t\epsilon \log L$ & [Thm. \ref{thm:deletion:asymptotic}] \\ \hline
			$(0,1,1)_\sub$ &$2\log L $ & [Thm. \ref{thm:gv:s:t:e:sub}] & $2L$ & [Const. \ref{con:constant:weight}] & $\log (ML)$ & [Thm. \ref{thm:substitutions:asymptotic}] \\ 
			$(0,1,1)_{\del}$ & $2\log L $ & [Thm. \ref{thm:gv:s:t:e:del}] & $\log L$ & [Const. \ref{con:single:insertion}] &$\log L$ & [Thm. \ref{thm:deletion:asymptotic}] \\ \hline
			$(0,M,\epsilon)_\sub$ & $2M\epsilon \log L$ & [Thm. \ref{thm:gv:s:t:e:sub}] & $M\epsilon\log L$ & [Const. \ref{con:hamming}] & $M \epsilon \log L$ & [Thm. \ref{thm:substitution:scaling}] \\ 
			$(0,M,1)_{\del}$ & $2M \log L$ & [Thm. \ref{thm:gv:s:t:e:del}] & $M\log L$ & [Const. \ref{con:multiple:insertion}] & $M\log L$ & [Thm. \ref{thm:deletion:scaling}] \\ \specialrule{.8pt}{0pt}{0pt}
	\end{tabular}}
	\label{tab:construction}
\end{table*}

\section{Gilbert-Varshamov Bounds} \label{sec:gv:bounds}
We start by deriving Gilbert-Varshamov lower bounds on the size (equivalently, upper bounds on the redundancy) of optimal $(s,t,\epsilon)_\sub$ and $(s,t,\epsilon)_\del$-correcting codes. An important entity for the derivation of the Gilbert-Varshamov bounds is the set of words $\tilde{\mathcal{S}} \in \mathcal{X}_M^L$, which have intersecting error balls with some $\mathcal{S} \in \mathcal{X}_M^L$. It is defined as follows.
\begin{defn}
For a set $\mathcal{S} \in \mathcal{X}_M^L$, we denote by $V_{s,t,\epsilon}^\ET(\mathcal{S})$ the set of all sets $\tilde{\mathcal{S}}\in\mathcal{X}_M^L$, which have intersecting error balls $B^\ET_{s,t,\epsilon}(\cdot)$ with $\mathcal{S}$, that is, $$ V_{s,t,\epsilon}^\ET(\mathcal{S}) = \{\tilde{\mathcal{S}} \in \mathcal{X}_M^L: B_{s,t,\epsilon}^\ET(\mathcal{S}) \cap B_{s,t,\epsilon}^\ET(\tilde{\mathcal{S}}) \neq \emptyset \}. $$
Hereby, $|V_{s,t,\epsilon}^\ET(\mathcal{S})|$ is called the degree of $\mathcal{S}$. The average degree of all sets is denoted by
$$ \avg{V_{s,t,\epsilon}^\ET} = \frac{1}{\binom{2^L}{M}} \sum_{\mathcal{S} \in \mathcal{X}_M^L} |V_{s,t,\epsilon}^\ET(\mathcal{S})|.$$
\end{defn}
The generalized Gilbert-Varshamov bound (cf. \cite{GF93,T97}) is based on a graph-theoretic representation of an error-correcting code. We will use this representation to find the generalized Gilbert-Varshamov bound for the DNA storage channel. Consider the simple graph  $\mathcal{G}$ with the set of vertices $\mathcal{X}_M^L$. Two vertices $\mathcal{S}_1, \mathcal{S}_2 \in \mathcal{X}_M^L$ are connected, if and only if $B_{s,t,\epsilon}^\ET(\mathcal{S}_1) \cap B_{s,t,\epsilon}^\ET(\mathcal{S}_2) = \emptyset$. Note that this definition is slightly different from \cite{GF93,T97} due to the lack of a distance measure in our case. By construction, a \textit{clique} in $\mathcal{G}$ (collection of vertices, where each pair of vertices is connected) is an $(s,t,\epsilon)_\ET$-correcting code. Now, it can directly be shown that the total number of edges $\mathcal{G}$ coincides with \cite[eq. (2)]{T97}. Analogously to \cite{T97}, it is therefore possible to establish a lower bound on the size of a clique in $\mathcal{G}$ (and therefore an $(s,t,\epsilon)_\ET$-correcting code).
\begin{thm}[cf. \cite{GF93,T97}] \label{thm:gen:gv:bound}
	There exists an $(s,t,\epsilon)_\ET$-correcting code $\Code \subseteq \mathcal{X}_M^L$ of size at least
	$$|\Code| \geq \frac{\binom{2^L}{M}}{ \avg{V_{s,t,\epsilon}^\ET}}.$$
\end{thm}
Such a code can be constructed by successively selecting words $\mathcal{S}^{(i)}$ with minimum degree from $\mathcal{X}_M^L$ as codewords and removing all words $V_{s,t,\epsilon}^\ET(\mathcal{S}^{(i)})$ as possible candidates for the succeeding codewords. Bounding the denominator in Theorem \ref{thm:gen:gv:bound} from above will be the main challenge in this section. We start by stating the bound for the case of an arbitrary number of errors per sequence.
\begin{thm} \label{thm:gv:s:t:bullet}
	There exists an $(s,t,\bullet)_\mathbb{L}$-correcting code $\Code \subseteq \mathcal{X}_M^L$ of cardinality at least
	$$ |\Code|  \geq \frac{\binom{2^L}{M}}{\binom{M}{s+2t} \binom{2^L}{s+2t} }.$$
	Hence, for fixed $s,t \in \mathbb{N}_0$ and fixed $0<\beta<1$, there exists an $(s,t,\bullet)_\mathbb{L}$-correcting code $\Code \subseteq \mathcal{X}_M^L$ with redundancy
	$$ r(\Code) \leq (s+2t)L + (s+2t) \log M - \log ((s+2t)!^2) +o(1), $$
	when $M \rightarrow \infty$ with $M = 2^{\beta L}$.
\end{thm}
\begin{proof}
	We will find an upper bound on $\avg{V_{s,t,\bullet}^\mathbb{L}}$ by  bounding $|V^\mathbb{L}_{s,t,\bullet}(\mathcal{S})|$ from above for all $\mathcal{S} \in \mathcal{X}_M^L$. In the following, let $\tilde{\mathcal{S}} \in V^\mathbb{L}_{s,t,\bullet}(\mathcal{S}) \subseteq \mathcal{X}_M^L$ be a set which has an intersecting error ball with $\mathcal{S}$. Start by observing that for any such $\tilde{S}$, there exists $\mathcal{S}'\in B_{s,t,\bullet}^\mathbb{L}(\mathcal{S})\cap B_{s,t,\bullet}^\mathbb{L}(\tilde{\mathcal{S}})$ with $|\mathcal{S}'| \leq M-s$, since $B_{s,t,\bullet}^\mathbb{L}(\mathcal{S}) \cap B_{s,t,\bullet}^\mathbb{L}(\tilde{\mathcal{S}}) \neq \emptyset$ and for all $\mathcal{S}'' \in B_{s,t,\bullet}^\mathbb{L}(\mathcal{S}) \cap B_{s,t,\bullet}^\mathbb{L}(\tilde{\mathcal{S}})$ with $|\mathcal{S}''|>M-s$ it is possible to construct $\mathcal{S}' \in B_{s,t,\bullet}^\mathbb{L}(\mathcal{S}) \cap B_{s,t,\bullet}^\mathbb{L}(\tilde{\mathcal{S}})$ with $|\mathcal{S}'| = M-s$ by removing any $|\mathcal{S}''|-M+s$ sequences from $\mathcal{S}''$. By Definition \ref{def:dna:channel:error:ball}, $|\mathcal{S} \cap \mathcal{S}'| \geq M-s-t$ and also $|\tilde{\mathcal{S}} \cap \mathcal{S}'| \geq M-s-t$. Further, for any such $\mathcal{S}'$,
	\begin{align*}
		|\mathcal{S} \cap \tilde{\mathcal{S}}| &\geq |\mathcal{S} \cap \tilde{\mathcal{S}} \cap \mathcal{S}'| \overset{(a)}{\geq} |\mathcal{S} \cap \mathcal{S}'| + |\tilde{\mathcal{S}} \cap \mathcal{S}'| - |\mathcal{S}'|\\
		&\geq 2(M-s-t)-(M-s) = M-s-2t,
	\end{align*}
	where we used in $(a)$ that $|\mathcal{S} \cap \tilde{\mathcal{S}} \cap \mathcal{S}'| =   |\mathcal{S} \cap  \mathcal{S}'| + |\tilde{\mathcal{S}} \cap  \mathcal{S}'| -  |(\mathcal{S} \cup\tilde{\mathcal{S}}) \cap \mathcal{S}'| 
	\geq
	|\mathcal{S} \cap  \mathcal{S}'| + |\tilde{\mathcal{S}} \cap  \mathcal{S}'| -  |\mathcal{S}'| $ (for an illustration, refer to Fig. \ref{fig:illustration:proof:gv:thm}). Therefore, any $\tilde{\mathcal{S}}$ has an intersection of size at least $M-s-2t$ with $\mathcal{S}$. Note that for $2^L\geq M+s+2t$ this bound is tight, i.e., it is possible to find sets $\mathcal{S},\tilde{\mathcal{S}} \in \mathcal{X}_M^L$ with $B_{s,t,\bullet}^\mathbb{L}(\mathcal{S})\cap B_{s,t,\bullet}^\mathbb{L}(\tilde{\mathcal{S}}) \neq \emptyset$ and $\mathcal{S} \cap \tilde{\mathcal{S}} = M-s-2t$.  Each $\tilde{\mathcal{S}}$ can thus be constructed by removing $s+2t$ sequences from $\mathcal{S}$ and adding $s+2t$ arbitrary sequences. The total number of such choices is at most $\binom{M}{s+2t}\binom{2^L}{s+2t}$. The bound on the redundancy follows from Definition \ref{def:redundancy} and the fact that for any fixed $a \in \mathbb{N}_0$, $\log \binom{M}{a} = a \log M - \log a! + o(1)$.
\newcommand{\squarelines}{\vcenter{\hbox{\begin{tikzpicture} 
			\draw[pattern=north east lines,pattern color=black!50!] (-.5em,-.5em) rectangle (.5em,.5em);
			\end{tikzpicture}}}}
\newcommand{\squarelinesrev}{\vcenter{\hbox{\begin{tikzpicture} 
			\draw[pattern=north west lines,pattern color=black!50!] (-.5em,-.5em) rectangle (.5em,.5em);
			\end{tikzpicture}}}}
\newcommand{\squaregray}{\vcenter{\hbox{\begin{tikzpicture} 
			\draw[pattern=crosshatch,pattern color=black!50!] (-.5em,-.5em) rectangle (.5em,.5em);
			\end{tikzpicture}}}}
	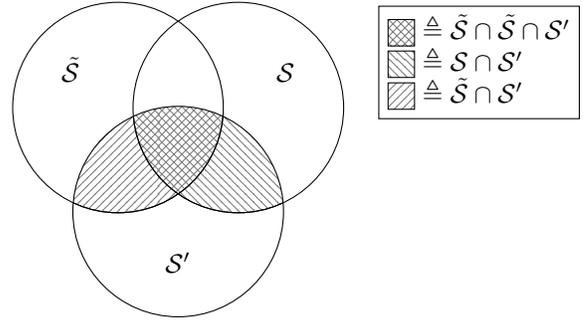
\begin{figure}
		\def\firstellipse{(1*0.8, 0) circle (1.4cm)}
		\def\secondellipse{(-1*0.8, 0) circle (1.4cm)}
		\def\thirdellipse{(0, -1.73*0.8) circle (1.4cm)}
		
		\centering
		\begin{tikzpicture}
			\begin{scope}[even odd rule]
				\clip \thirdellipse;
				\draw[pattern=north east lines,pattern color=black!50!] \secondellipse;
			\end{scope}
			\begin{scope}[even odd rule]
				\clip \thirdellipse;
				\draw[pattern=north west lines,pattern color=black!50!] \firstellipse;
			\end{scope}
			\draw \firstellipse node[label={[label distance=.3cm]30:$\mathcal{S}$}] {};
			\draw \secondellipse node[label={[label distance=.3cm]150:$\tilde{\mathcal{S}}$}] {};
			\draw \thirdellipse node[label={[label distance=.3cm]270:$\mathcal{S}'$}] {};
			
			\node[draw,align=left] at (4,.6)(legend) {$\squaregray \triangleq \tilde{\mathcal{S}} \cap \tilde{\mathcal{S}}\cap \mathcal{S}'$\\$\squarelinesrev \triangleq \mathcal{S} \cap \mathcal{S}'$\\$\squarelines \triangleq \tilde{\mathcal{S}} \cap \mathcal{S}'$};
		\end{tikzpicture}%
		\caption{Illustration for the proof of Theorem \ref{thm:gv:s:t:bullet}}
		\label{fig:illustration:proof:gv:thm}
	\end{figure}
\end{proof}
In a similar fashion, we will now establish the existence of a code for the case of a loss of $s$ sequences and a fixed number of $\epsilon$ substitution errors in $t$ sequences.
\begin{thm} \label{thm:gv:s:t:e:sub}
	There exists an $(s,t,\epsilon)_\sub$-correcting code $\Code \subseteq \mathcal{X}_M^L$ with cardinality at least
	$$ |\Code| \geq \frac{\binom{2^L}{M}}{\binom{M}{s,t} \binom{M+t-1}{t} \binom{2^L}{s} B_\epsilon^\sub(L)^{2t}}. $$
	Hence, for fixed $s,t,\epsilon \in \mathbb{N}_0$ and fixed $0<\beta<1$, there exists an $(s,t,\epsilon)_\sub$-correcting code $\Code \subseteq \mathcal{X}_M^L$ with redundancy
	$$ r(\Code) \leq sL + (s+2t) \log M +2t\epsilon \log L- \log (s!^2t!\epsilon!^{2t}) +o(1), $$
	when $M \rightarrow \infty$ with $M = 2^{\beta L}$.
\end{thm}
\begin{proof}
	We will find an upper bound on $|V^\sub_{s,t,\epsilon}(\mathcal{S})|$ for all $\mathcal{S} \in \mathcal{X}_M^L$. Let $\mathcal{S}' \in B_{s,t,\epsilon}^\sub(\mathcal{S})$ with $|\mathcal{S}'|\leq M-s$. The number of such elements $\mathcal{S}'$ is at most $\binom{M}{s,t} B_\epsilon^\sub(L)^t$, as we can choose $s$ sequences to be lost, $t$ sequences to be erroneous and there are $B_\epsilon^\sub(L)$ error patterns for each erroneous sequence. Given $\mathcal{S}' \in B_{s,t,\epsilon}^\sub(\mathcal{S})$, we construct possible $\tilde{\mathcal{S}}$ with $\mathcal{S}' \in B_{s,t,\epsilon}^\sub(\tilde{\mathcal{S}})$ as follows. For each of the $t$ erroneous sequences it is possible to either add $\epsilon$ errors to a sequence $\ve{x} \in \mathcal{S}'$ or to create a new sequence inside the error ball $B_\epsilon^\sub(\ve{x})$. There are $\binom{M+t-1}{t}B_\epsilon^\sub(L)^{t}$ possible error patterns for this procedure. Finally, the $s$ lost sequences can be arbitrary sequences $\ve{x} \in \Sigma_2^L$, and there are at most $\binom{2^L}{s}$ choices for these sequences. Thus,
	$$ |V^\sub_{s,t,\epsilon}(\mathcal{S})| \leq \binom{M}{s,t} B_\epsilon^\sub(L)^t \binom{M+t-1}{t} B_\epsilon^\sub(L)^t \binom{2^L}{s}. $$
	Applying Theorem \ref{thm:gen:gv:bound} and using the definition of the redundancy directly yields the bounds of the theorem.
\end{proof}
For the case of deletion errors, we slightly adapt our arguments since the size of the deletion sphere is non-uniform \cite{Lev66}. As stated in Theorem \ref{thm:gen:gv:bound}, it is sufficient to find an upper bound on the average degree $\avg{V_{s,t,\epsilon}^\del}$.
\begin{defn}
	The average of the $t$-th power of the deletion sphere size $|S_\epsilon^\del(\ve{x})|$ over all $\ve{x} \in \Sigma_2^L$ is defined to be
	$$\avg{S_\epsilon^{\del,t}} = \frac{1}{2^L} \sum_{\ve{x} \in \Sigma_2^L} |S_\epsilon^\del(\ve{x})|^t. $$
\end{defn}
Based on this definition we formulate the following theorem about the existence of $(s,t,\epsilon)_\del$-correcting codes.
\begin{thm} \label{thm:gv:s:t:e:del}
	There exists an $(s,t,\epsilon)_\del$-correcting code $\Code \subseteq \mathcal{X}_M^L$ with cardinality at least
	$$ |\Code| \geq \frac{\binom{2^L}{M}}{\binom{M}{s,t} \binom{2^L}{s} B_\epsilon^\sub(L)^t \avg{S_\epsilon^{\del,t}}}. $$
	Hence, for fixed $s,t,\epsilon \in \mathbb{N}_0$ and fixed $0<\beta<1$, there exists an $(s,t,\epsilon)_\del$-correcting code $\Code \subseteq \mathcal{X}_M^L$ with redundancy
	$$ r(\Code) \leq sL + (s+t) \log M +2t\epsilon \log L -t\epsilon -\log (s!^2t!^2\epsilon!^{2t}) + o(1), $$
	when $M \rightarrow \infty$ with $M = 2^{\beta L}$.
\end{thm}
\begin{proof}
	We will derive an upper bound on $\avg{V_{s,t,\epsilon}^\del}$. The number of elements in $\mathcal{S}' \in B_{s,t,\epsilon}^\del(\mathcal{S})$ after a loss of exactly $s$ sequences and $\epsilon$ deletions in $t$ sequences is at most
	$$ |B_{s,t,\epsilon}^\del(\mathcal{S})| \leq \sum_{\SE\subseteq \mathcal{S}, |\SE|=t}\prod_{\ve{x} \in \SE} |S_\epsilon^\del(\ve{x})| \binom{M-t}{s}. $$
	This can be illustrated by the following consideration. First, fix $\SE\subseteq \mathcal{S}$ with $|\SE| = t$. There are $|S_\epsilon^\del(\ve{x})|$ possible error patterns for each $\ve{x} \in \SE$ and $\binom{M-t}{s}$ choices of $s$ lost sequences among the remaining $M-t$ error-free sequences. Summing over all possible choices $\SE \subseteq\mathcal{S}$ of erroneous sequences yields the bound. Then, for each such set $\mathcal{S}'$, there are at most $ \binom{2^L}{s}S_\epsilon^\ins(L-\epsilon)^t$ sets $\tilde{\mathcal{S}}$ with $\mathcal{S}' \in B_{s,t,\epsilon}^\del(\tilde{\mathcal{S}}) \neq \emptyset$. This is because each erroneous sequence $\ve{x}' \in \mathcal{S}'$ has length $L-\epsilon$ and requires $\epsilon$ insertions to become a sequence of length $L$. The $s$ lost sequences can be arbitrary words in $\tilde{\mathcal{S}}$ and therefore
	$$|V_{s,t,\epsilon}^\del(\mathcal{S})| \leq |B_{s,t,\epsilon}^\del(\mathcal{S})| \binom{2^L}{s} S_\epsilon^\ins(L-\epsilon)^t  .$$
	Taking the average of $|B_{s,t,\epsilon}^\del(\mathcal{S})|$ over all sets $\mathcal{S} \in \mathcal{X}_M^L$ yields
	\begin{align*}
		\sum_{\mathcal{S} \in \mathcal{X}_M^L}& \frac{|B_{s,t,\epsilon}^\del(\mathcal{S})|}{\binom{2^L}{M}}  \leq \frac{\binom{M-t}{s}}{\binom{2^L}{M}} \sum_{\mathcal{S} \in \mathcal{X}_M^L}  \sum_{\SE\subseteq \mathcal{S}, |\SE|=t}\prod_{\ve{x} \in \SE} |S_\epsilon^\del(\ve{x})|\\   &\overset{(a)}{=}  \frac{\binom{M}{s,t}}{\binom{2^L}{t}} \sum_{\SE \in \mathcal{X}_t^L} \prod_{\ve{x} \in \SE} |S_\epsilon^\del(\ve{x})| 
		 \overset{(b)}{\leq} \frac{\binom{M}{s,t}}{\binom{2^L}{t}} \sum_{\SE \in \mathcal{X}_t^L} \sum_{\ve{x} \in \SE} \frac{|S_\epsilon^\del(\ve{x})|^t}{t}\\ & \overset{(c)}{=} \binom{M}{s,t} \avg{S_\epsilon^{\del,t}}.
	\end{align*}
	Here, for equality $(a)$ we used that each set $\SE$ with $|\SE| = t$ is contained in exactly $\binom{2^L-t}{M-t}$ sets $\mathcal{S} \in \mathcal{X}_M^L$. It follows from the combination of the arithmetic-geometric mean inequality and Jensen inequality that for any non-negative $a_1,\dots,a_t\geq 0$ it holds that $a_1\cdot \ldots \cdot a_t \leq \frac1t(a_1^t+\ldots+a_t^t)$, which has been used in inequality $(b)$. Equality $(c)$ follows from the fact that each $\ve{x} \in \Sigma_2^L$ is contained in $\binom{2^L-1}{t-1}$ sets $\SE \in \mathcal{X}_t^L$. It is known \cite{Lev66} that $|S_\epsilon^\del(\ve{x})| \leq \binom{||\ve{x}||+\epsilon-1}{\epsilon} \leq \frac{(||\ve{x}||+\epsilon-1)^{\epsilon}}{\epsilon!}$, which results in
	\begin{align*}
		\avg{S_\epsilon^{\del,t}} &\leq \frac{1}{2^L} \sum_{\ve{x} \in \Sigma_2^L} \frac{(||\ve{x}||+\epsilon-1)^{t\epsilon}}{\epsilon!^t}  \\ &\overset{(a)}{=} \frac{1}{\epsilon!^t} \sum_{i=0}^{L-1} \frac{\binom{L-1}{i}(i+\epsilon)^{t\epsilon}}{2^{L-1}}
		 \overset{(b)}{\lesssim} \frac{1}{\epsilon!^t} \left(\frac{L}{2}\right)^{\epsilon t}.
	\end{align*}
	In equality $(a)$ it has been used that the number of words $\ve{x}\in \Sigma_2^L$ with $||\ve{x}||=i$ is $2\binom{L-1}{i-1}$. For inequality $(b)$, we identify the sum as the decentralized moment of a binomial distribution with $L-1$ trials and success probability $\frac12$ and use \cite[eq. (4.10)]{K08} for the asymptotic behavior, when $L \rightarrow \infty$.
\end{proof}
\section{Sphere-Packing Bounds} \label{sphere:packing:bounds}
A well-known method to find upper bounds on the cardinality of error-correcting codes is the sphere-packing bound. In this section we  derive sphere-packing bounds for $(s,t,\epsilon)_\ET$-correcting codes. These bounds directly imply lower bounds on the redundancy of such codes. One particular observation of the considered DNA storage channel is that it is non-uniform, i.e. the sizes of the error balls $B^{\ET}_{s,t,\epsilon}(\mathcal{S})$ depend on the channel input $\mathcal{S}$ for all types of errors $\ET$, which hinders the computation of sphere packing bounds. A practical method to find sphere packing bounds for non-uniform error balls is the generalized sphere packing bound \cite{FVY15,KK13}. However, due to the complex expressions of the error ball sizes, this method does not yield tractable expressions for the considered channel. Another possibility is to derive the sphere packing bound by finding an upper bound on the error ball size, which we will do in Section \ref{ss:non:asymptotic:bounds}. We will also show that for large $M$ most of the error balls have a similar size, which allows to formulate tighter asymptotic sphere packing bounds in Sections \ref{ss:asmyptotic:boubds:substitutions} and \ref{ss:asmyptotic:boubds:deletions}. Note that together with the lower bounds on the achievable size of $(s,t,\epsilon)_\ET$-correcting codes from the previous section and concrete code constructions in Section \ref{sec:const}, it can be shown that the sphere packing bounds are asymptotically tight for many channel parameters and provide important insights into the nature of the DNA channel.
\subsection{Non-Asymptotic Bounds} \label{ss:non:asymptotic:bounds}
We start by finding an upper bound for $(s,t,\bullet)_\mathbb{L}$-correcting codes, which depicts the case of a loss of $s$ sequences and an arbitrary number of edit errors in each of $t$ erroneous sequences.
\begin{thm} \label{thm:bound:arbitrary:error}
	The cardinality of any $(s,t,\bullet)_\mathbb{L}$-correcting code $\Code \subseteq \mathcal{X}_M^L$ satisfies
	$$ |\Code| \leq \frac{\binom{2^L}{M-s}}{\binom{M}{t+s} \binom{2^L-M}{t}}. $$
	In particular, the redundancy of any $(s,t,\bullet)_\mathbb{L}$-correcting code $\Code \subseteq \mathcal{X}_M^L$ is therefore at least
	$$ r(\Code) \geq (s+t) \log (2^L-M-t) + t\log (M-s-t) - \log (t!(s+t)!).$$
\end{thm}
\begin{proof}
	We prove the theorem by finding a subset of $B^\mathbb{L}_{s,t,\bullet}(\mathcal{S})$, which gives a lower bound on the sphere size $|B^\mathbb{L}_{s,t,\bullet}(\mathcal{S})|$ for all $\mathcal{S} \in \mathcal{X}_M^L$. Let $\mathcal{S}' \in B^\mathbb{L}_{s,t,\bullet}(\mathcal{S}) \cap \Sigma_2^L$ denote an element from the error ball of $\mathcal{S}$, which contains only sequences of length $L$ and let $\SG, \SE'$ denote the corresponding error-free, respectively erroneous outcomes of the sequences, i.e. $\mathcal{S}' = \SG \cup \SE'$, according to Definition \ref{def:dna:channel:error:ball}. We construct such distinct $\mathcal{S}'$ in the following way. Choose $M-s-t$ error-free sequences $\SG \subseteq \mathcal{S}$ and choose the $t$ erroneous sequences in $\SE'$ to be distinct elements out of the $2^L-M$ sequences in $\Sigma_2^L \setminus \mathcal{S}$ and let $\mathcal{S}' = \SG \cup \SE'$. For any such $\SG \subseteq \mathcal{S}$ and $\SE' \subseteq \Sigma_2^L\setminus \mathcal{S}$ one obtains a unique element from the error ball $B^\mathbb{L}_{s,t,\bullet}(\mathcal{S})$, since $\mathcal{S}' = \SG \cup \SE'$ and $\SG, \SE'$ are both subsets of two distinct sets. There are in total $\binom{M}{s+t}$ ways to choose the set $\SG$ and $\binom{2^L-M}{t}$ ways to choose $\SE'$ and thus $|B^\mathbb{L}_{s,t,\bullet}(\mathcal{S})| \geq \binom{M}{s+t}\binom{2^L-M}{t}$. All such constructed received sets have $|\mathcal{S}'|=|\SG| + |\SE'| = M-s$ sequences of length $L$ and therefore, we obtain by a sphere packing argument, that any $(s,t,\bullet)_\mathbb{L}$-correcting code $\Code$ satisfies
	$$ |\Code| \leq \frac{\binom{2^L}{M-s}}{\binom{M}{t+s} \binom{2^L-M}{t}}. $$
	Therefore, the redundancy is at least
	\begin{align*}
	r(\Code) =& \log \binom{2^L}{M} - \log|\Code| \\ \geq& \log \frac{(2^L-M+s)!(M-s)!}{(2^L-M-t)!(M-s-t)!(s+t)!t!} \\
	\geq& (s+t) \log (2^L-M-t) +t\log (M-t-s) \\&\qquad\qquad\qquad\qquad\qquad\;\;- \log(t!(s+t)!).
	\end{align*}
	\vspace{-.2cm}
\end{proof}
This non-asymptotic bound directly implies an asymptotic bound, when $M\rightarrow \infty$ and $M = 2^{\beta L}$ for fixed $0<\beta <1$.
\begin{corollary} \label{cor:bound:red:arbitrary:error}
	For fixed $s,t \in \mathbb{N}_0$ and fixed $0<\beta<1$, the redundancy of any $(s,t,\bullet)_\mathbb{L}$-correcting code $\Code \subseteq \mathcal{X}_M^L$ is asymptotically at least
	$$ r(\Code) \geq (s+t)L+t\log M - \log(t!(s+t)!) + o(1), $$
	when $M \rightarrow \infty$ and $M = 2^{\beta L}$. Further, for any fixed $\sigma, \tau$ with $\sigma>0$,  $\tau>0$ and $\sigma+\tau<1$, the redundancy of any $(\sigma M,\tau M,\bullet)_\mathbb{L}$-correcting code $\Code \subseteq \mathcal{X}_M^L$ satisfies
	$$ r(\Code) \geq (\sigma+\tau)M(L-\log M + \log \e) + MH(\sigma+\tau) + o(M), $$ 
	where $H(p) = -p\log p - (1-p)\log(1-p)$ is the binary entropy function.
\end{corollary}
This result is particularly interesting, due to the following consideration. Both lost sequences and erroneous sequences do not carry any useful information, since the erroneous sequences can be distorted by an arbitrary number of errors. However, unlike the lost sequence, the erroneous sequence cannot directly be detected by the decoder and therefore, compared to a loss of sequence, requires additional redundancy of roughly $\log M$ bits to be corrected. This result is analogous to the case of standard binary substitution-correcting block-codes of length $n$, where erasures require a redundancy of only a single symbol, and errors require roughly $\log n$ symbols of redundancy to be corrected. This analogy becomes particularly visible when sequences are indexed and protected by a standard substitution-correcting code, similarly to Construction \ref{con:index:rs} (see Section \ref{ss:index:based:construction:MDS}), but also holds for the general case of any $(s,t,\bullet)_\mathbb{L}$-correcting code. However, this seems to be not the case, when the number of lost sequences and erroneous sequences scales with $M$, since the redundancy only depends on $\sigma+\tau$.

In the following, we find code size upper bounds for the case of having a combination of a loss of $s$ sequences and only $\epsilon$ insertion errors inside $t$ sequences. We start by defining a quantity that will be useful for the formulation of the bound.
\begin{defn}
	The largest intersection of two $\epsilon$-insertion spheres of any two distinct words $\ve{x},\ve{y} \in \Sigma_2^L$ is denoted by
	$$N_\epsilon^\ins(L) = \underset{\ve{x},\ve{y} \in \Sigma_2^L}{\max} |S_\epsilon^\ins(\ve{x}) \cap S_\epsilon^\ins(\ve{y})|.$$
\end{defn}
Note that from \cite{Lev01} it is known that $N_\epsilon^\ins(L) = \sum_{i=0}^{\epsilon-1}\binom{L+\epsilon}{i}(1-(-1)^{\epsilon-i})$. The sphere packing bound is derived in the following theorem.
\begin{thm} \label{thm:bound:insertion} The cardinality of any $(s,t,\epsilon)_\ins$-correcting code $\Code \subseteq \mathcal{X}_M^L$ satisfies
	$$|\Code| \leq \frac{\binom{2^L}{M-s-t}\binom{2^{L+\epsilon}}{t}}{\binom{M}{s,t}\prod_{i=0}^{t-1}(S_\epsilon^\ins(L) - (s+i)N_\epsilon^\ins(L)}.$$
\end{thm}
\begin{proof}
	We prove the theorem by bounding the error ball size $|B^\ins_{s,t,\epsilon}(\mathcal{S})|$ from below for all $\mathcal{S}$, which yields an upper bound on the cardinality of $(s,t,\epsilon)_\ins$-correcting codes by a sphere packing argument. Distinct elements $S' \in B^\ins_{s,t,\epsilon}(\mathcal{S})$ of the error ball can be found in the following way. First, choose two distinct sets $\SL,\SE=\{\ve{x}_{e_1}, \dots, \ve{x}_{e_t} \} \subseteq \mathcal{S}$ with $|\SL| = s$ and $|\SE| = t$. Further choose the set of erroneous sequences $\SE' = \{\ve{x}_1', \dots, \ve{x}_t' \}$ such that
	$$ \ve{x}_i' \in S_\epsilon^\ins(\ve{x}_{e_i}) \bigg\backslash \Bigg(\bigcup_{\ve{y} \in \mathcal{P}_i} S_\epsilon^\ins(\ve{y})\Bigg)$$
	as illustrated in Fig. \ref{fig:illustration:erroneous:sequences:insertion}, where $\mathcal{P}_i = \SL \cup \{\ve{x}_{e_1}, \dots,\ve{x}_{e_{i-1}}  \}$. The received set $\mathcal{S}'$ is then constructed by $\mathcal{S}' = \SG\cup \SE'$, where $\SG = \mathcal{S}\setminus (\SL\cup\SE)$ are the error-free sequences, as in Definition \ref{def:dna:channel:error:ball}. We will show that each choice $\SL,\SE,\SE'$ leads to a unique element in $B^\ins_{s,t,\epsilon}(\mathcal{S})$. Denote by $\SL,\SE,\SE'$ and $\tilde{\SL},\tilde{\SE},\tilde{\SE}'$ two different choices and let $\mathcal{S}'$ and $\tilde{\mathcal{S}}'$ be the corresponding received sets. If $\SL \cup \SE \neq \tilde{\SL} \cup \mathcal{\tilde{F}}$, it directly follows that $\mathcal{S}'\neq \tilde{\mathcal{S}}'$, since the error-free sequences are different. However, if $\SL \cup \SE = \tilde{\SL} \cup \mathcal{\tilde{F}}$, it follows that $\SE' \neq \tilde{\SE}'$ due to the choice of the sequences in the set $\SE'$. Therefore, two different choices of the sets $\SL,\SE,\SE'$ yield different elements in $B^\ins_{s,t,\epsilon}(\mathcal{S})$. The number of possible sets $\SL,\SE$ is $\binom{M}{s,t}$. For each $\ve{x}_{e_i} \in \SE$, we have at least $S_\epsilon^\ins(L) - (s+i)N_\epsilon^\ins(L)$ possibilities to choose the erroneous outcome $\ve{x}_i'$, since there are $S_\epsilon^\ins(L)$ sequences in $S_\epsilon^\ins(\ve{x}_{e_i})$ and at most $(s+i)N_\epsilon^\ins(L)$ of them are in common with elements of the insertion spheres of $\mathcal{P}_i$. Hence, in total, there are $\binom{M}{s,t}\prod_{i=0}^{t-1}(S_\epsilon^\ins(L) - (s+i)N_\epsilon^\ins(L))$ ways to choose $\SL, \SE, \SE'$ and therefore $|B^\ins_{s,t,\epsilon}(\mathcal{S})|\geq\binom{M}{s,t}\prod_{i=0}^{t-1}(S_\epsilon^\ins(L)-(s+i)N_\epsilon^\ins(L))$ for all $\mathcal{S} \in \mathcal{X}_M^L$. Each such created received set $\mathcal{S}'$ consists of $M-s-t$ sequences of length $L$ and $t$ sequences of length $L+\epsilon$. There are in total $\binom{2^L}{M-s-t}\binom{2^{L+\epsilon}}{t}$ such sets, which yields the theorem by a sphere packing argument.
\end{proof}
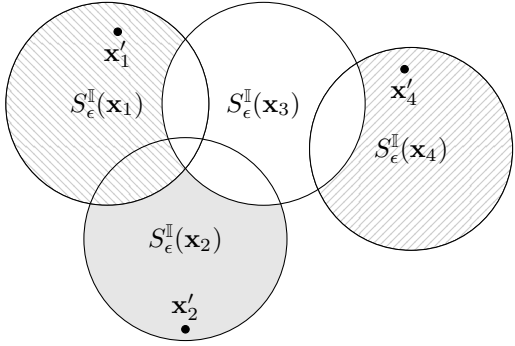
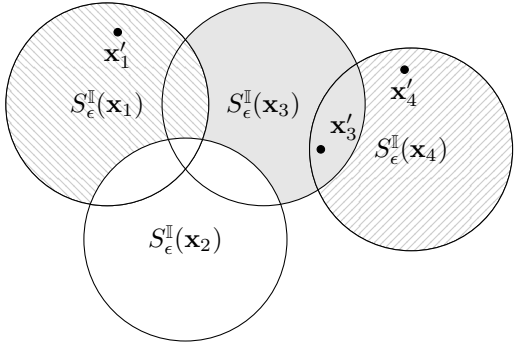
\begin{figure}
	\def\firstcircle{(150:1.2cm) circle (1.35cm)}
	\def\secondcircle{(270:1.2cm) circle (1.35cm)}
	\def\thirdcircle{(30:1.2cm) circle (1.35cm)}
	\def\fourthcircle{(0:3.0cm) circle (1.35cm)}
	
	\begin{minipage}{\linewidth}
		\centering
		\begin{tikzpicture}
		\begin{scope}[even odd rule]
		\clip \thirdcircle \firstcircle;
		\draw[pattern=north west lines,pattern color=black!20!] \firstcircle;
		\end{scope}
		\begin{scope}[even odd rule]
		\clip \firstcircle \secondcircle;
		\clip \thirdcircle \secondcircle;
		\fill[fill=black!10!] \secondcircle;
		\end{scope}
		\begin{scope}[even odd rule]
		\clip \thirdcircle \fourthcircle;
		\draw[pattern=north east lines,pattern color=black!20!] \fourthcircle;
		\end{scope}
		\draw \firstcircle node {$S_\epsilon^\ins(\ve{x}_{1})$};
		\draw \secondcircle node {$S_\epsilon^\ins(\ve{x}_{2})$};
		\draw \thirdcircle node{$S_\epsilon^\ins(\ve{x}_3)$};
		\draw \fourthcircle node{$S_\epsilon^\ins(\ve{x}_4)$};
		
		\draw[fill] (120:1.8) circle (0.05cm) node [below] {$\ve{x}_1'$};
		\draw[fill] (270:2.4) circle (0.05cm) node [above] {$\ve{x}_2'$};
		\draw[fill] (20:3.1) circle (0.05cm) node [below] {$\ve{x}_4'$};
		\end{tikzpicture}%
		\subcaption{Exemplary case: $\ve{x}_1, \ve{x}_2, \ve{x}_4 \in \SE$, and $\ve{x}_3 \in \SL$}
	\end{minipage}\vspace{.2cm}
	\begin{minipage}{\linewidth}
		\centering
		\begin{tikzpicture}
		\begin{scope}[even odd rule]
		\clip \secondcircle \firstcircle;
		\draw[pattern=north west lines,pattern color=black!20!] \firstcircle;
		\end{scope}
		\begin{scope}[even odd rule]
		\clip \firstcircle \thirdcircle;
		\clip \secondcircle \thirdcircle;
		\fill[fill=black!10!] \thirdcircle;
		\end{scope}
		\begin{scope}[even odd rule]
		\clip \thirdcircle \fourthcircle;
		\draw[pattern=north east lines,pattern color=black!20!] \fourthcircle;
		\end{scope}
		\draw \firstcircle node {$S_\epsilon^\ins(\ve{x}_{1})$};
		\draw \secondcircle node {$S_\epsilon^\ins(\ve{x}_{2})$};
		\draw \thirdcircle node{$S_\epsilon^\ins(\ve{x}_3)$};
		\draw \fourthcircle node{$S_\epsilon^\ins(\ve{x}_4)$};
		
		\draw[fill] (120:1.8) circle (0.05) node [below] {$\ve{x}_1'$};
		\draw[fill] (0:1.8) circle (0.05) node [above right] {$\ve{x}_3'$};
		\draw[fill] (20:3.1) circle (0.05) node [below] {$\ve{x}_4'$};
		\end{tikzpicture}%
		\subcaption{Exemplary case: $\ve{x}_1, \ve{x}_3, \ve{x}_4 \in \SE$, and $\ve{x}_2 \in \SL$}
	\end{minipage}
	\caption{Illustration for the choice of $\SE'$ in the proof of Theorem \ref{thm:bound:insertion}. The erroneous outcomes are chosen out of the corresponding error spheres, which are highlighted in gray.}
	\label{fig:illustration:erroneous:sequences:insertion}
\end{figure}
Note that, Theorem \ref{thm:bound:insertion} provides a valid upper bound for any parameter $M,L,s,t,\epsilon$. For the case of deletion errors or combinations of insertions and deletions, formulating a sphere packing bound based on the minimum error ball size yields a weak bound, since the minimum deletion ball size is $|B_\epsilon^\del(\ve{0})| = \epsilon+1$. Therefore, a conservative analysis similar to Theorem \ref{thm:bound:insertion} would yield unsatisfactory results. However, an asymptotic analysis, which yields asymptotically tighter bounds is possible, as we will see in Theorem \ref{thm:deletion:asymptotic}. 
\subsection{Asymptotic Bounds for Substitution Errors} \label{ss:asmyptotic:boubds:substitutions}
We now derive asymptotic sphere packing bounds for large numbers of sequences $M$ on the code size for $(s,t,\epsilon)_\sub$-correcting codes, which depicts the case of only substitution errors inside the sequences. As discussed before, the error ball sizes depend on the center $\mathcal{S}$. However, as it turns out, asymptotically the error balls have similar sizes. We will start by finding a lower bound on the error ball size for a set $\mathcal{S}$.
\begin{lemma} \label{lemma:substitutions:subset}
	Let $\mathcal{Y} \subseteq \mathcal{S} \in \mathcal{X}_M^L$ be an $\epsilon$-substitution-correcting code, i.e. $B_\epsilon^\sub(\ve{y}_1) \cap B_\epsilon^\sub(\ve{y}_2) = \emptyset$ for all $\ve{y}_1, \ve{y}_2 \in \mathcal{Y}$ and $\ve{y}_1 \neq \ve{y}_2$. Further, let $s+t\leq |\mathcal{Y}|$. Then,
	$$|B_{s,t,\epsilon}^\sub(\mathcal{S})| \geq \binom{|\mathcal{Y}|}{s,t}  \left(B_{\epsilon}^\sub(L)-1\right)^{t}.$$
\end{lemma}
\begin{proof}
	A lower bound for $|B_{s,t,\epsilon}^\sub(\mathcal{S})|$ will be proven by identifying and counting specific patterns of a loss of sequences and errors in sequences that lead to distinct channel outputs $\mathcal{S}' \in B_{s,t,\epsilon}^\sub(\mathcal{S})$. Throughout this proof, we impose a lexicographic ordering onto the sequences in $\Sigma_2^L$, which means that, writing $\mathcal{A}= \{\ve{a}_1, \dots, \ve{a}_{|\mathcal{A}|} \}$ for any set $\mathcal{A} \subseteq \Sigma_2^L$ uniquely determines each element $\ve{a}_i$. The sets of stored sequences in the error balls around the elements in $\mathcal{Y}$ are denoted by $\mathcal{Y}_i = \mathcal{S} \cap B_\epsilon^{\sub}(\ve{y}_i)$. Similarly, the sets of received sequences in these error balls are $\mathcal{Y}'_i = \mathcal{S}' \cap B_\epsilon^{\sub}(\ve{y}_i)$. Note that the sets $B_\epsilon^{\sub}(\ve{y}_i)$ and thus also the sets $\mathcal{Y}_i$ are distinct, since $\mathcal{Y}$ is an $\epsilon$-substitution-correcting code. We further define the selector function for sequences $\ve{a},\ve{b}, \ve{x} \in \Sigma_2^L$ as
	$$ \mathbb{I}_\ve{x}^\mathcal{S}(\ve{a},\ve{b}) = \left\{\begin{array}{ll}
	\ve{a}, & \text{if } \ve{x} \notin \mathcal{S} \\
	\ve{b}, & \text{otherwise}
	\end{array} \right. . $$
	The distinct channel outputs $\mathcal{S}' \in B_{s,t,\epsilon}^\sub(\mathcal{S})$ are obtained in the following manner. First, choose two distinct sets $\SL  \subseteq \mathcal{Y}$ with $|\SL| = s$ and $\SE = \{\ve{y}_{e_1}, \dots, \ve{y}_{e_t} \} \subseteq \mathcal{Y}$ with $|\SE| = t$ and a collection of error vectors $\ve{E} = (\ve{e}_1, \ve{e}_2, \dots, \ve{e}_t)$, where $\ve{e}_j \in \Sigma_2^L$ are non-zero error vectors of weight at most $\epsilon$. We will show that for each choice of $\SL, \SE$, and $\ve{E}$ we obtain a unique point $\mathcal{S}' \in B_{s,t,\epsilon}^\sub(\mathcal{S})$ in the following manner. First, all sequences in $\SL$ are lost. Let $\ve{y}_i' \triangleq \ve{y}_{e_i} + \ve{e}_i$. The set $\SE$ of erroneous sequences is chosen as
	$$ \SE = \bigcup_{i=1}^t \left\{\mathbb{I}_{\ve{y}_i'}^\mathcal{S}(\ve{y}_{e_i}, \ve{y}_i')\right\}. $$
	In other words, if $\ve{y}_i' \notin \mathcal{S}$ we choose the sequence, which will be distorted by errors to be $\ve{y}_{e_i}$ and otherwise we choose it to be exactly $\ve{y}_i'$. The erroneous outcomes of the sequences in $\SE$ are now constructed by
	$$ \SE' = \bigcup_{i=1}^t \left\{\mathbb{I}_{\ve{y}_i'}^\mathcal{S}(\ve{y}_i', \ve{y}_{e_i})\right\}. $$
	That is if $\ve{y}_i' \notin \mathcal{S}$, we have $\ve{y}_{e_i} \in \SE$ and we add $\ve{e}_i$ to that sequence to obtain $\ve{y}_i'\in \SE'$. If $\ve{y}_i' \in \mathcal{S}$, $\ve{y}_i' \in \SE$ is the sequence which is distorted and we add $-\ve{e}_i$, resulting in $\ve{y}_{e_i} \in \SE'$. It is very important to note that by this choice of error patterns, the erroneous sequence $\ve{y}_i' \in B_\epsilon^\sub(\ve{y}_{e_i})$ and therefore will never be present in another error ball $B_\epsilon^\sub(\ve{y}), \ve{y} \in \mathcal{Y} \setminus \{\ve{y}_{e_i}\}$, since $\mathcal{Y}$ is an $\epsilon$-substitution-correcting code. The received set is now $\mathcal{S}' = \SG \cup \SE'$, where $\SG = \mathcal{S} \setminus (\SL\cup \SE)$ are the error-free sequences, as in Definition \ref{def:dna:channel:error:ball}. We will show now that two choices $\SL, \SE, \ve{E}$ and $\tilde{\SL}, \tilde{\SE}, \tilde{\ve{E}}$ yield different received sets $\mathcal{S}'$ and $\tilde{\mathcal{S}}'$, if (and only if) they differ in at least one of the components, i.e., $\SL \neq \tilde{\SL}$, $\SE \neq \tilde{\SE}$, or $\ve{E} \neq \tilde{\ve{E}}$. We distinguish between the following three different cases (visualized in Fig. \ref{fig:error:patterns}) and the resulting received parts $\mathcal{Y}_i'$
	\begin{itemize}
		\item $\ve{y}_i \in \mathcal{Y}\setminus (\SL\cup \SE): \mathcal{Y}_i' = \mathcal{Y}_i $,
		\item $\ve{y}_i \in \SL: \mathcal{Y}_i' = \mathcal{Y}_i \setminus \{\ve{y}_i\}$,
		\item \makebox[1.4cm][l]{$\ve{y}_i \in \SE:$} $\mathcal{Y}_i' = (\mathcal{Y}_i\setminus \{\ve{y}_i\}) \cup \{\ve{y}_i' \}$ or 
		\item[]\makebox[1.54cm][l]{}$\mathcal{Y}_i' = \mathcal{Y}_i \setminus \{\ve{x}\}$,
	\end{itemize}
	where $\ve{y}_i' \in B_\epsilon^\sub(\ve{y}_i) \setminus \mathcal{S}$ and $\ve{x} \in \mathcal{Y}_i \setminus \{\ve{y}_i\}$. By comparing the outputs $\mathcal{Y}_i'$ for these three cases, it is verified that for any two different cases, $\mathcal{Y}_i'$ can never be the same. Now, if $\SL \neq \tilde{\SL}$ there is at least one $i$ such that $\ve{y}_i \in \SL$ and $\ve{y}_i \notin \tilde{\SL}$ and if $\SE \neq \tilde{\SE}$ there is at least one $i$ such that $\ve{y}_i \in \SE$ and $\ve{y}_i \notin \tilde{\SE}$ and therefore it follows that $\mathcal{Y}_i' \neq \tilde{\mathcal{Y}}_i'$ and $\mathcal{S}' \neq \tilde{\mathcal{S}}'$. Further, if $\SL = \tilde{\SL}$ and $\SE = \tilde{\SE}$, but $\ve{E}\neq \tilde{\ve{E}}$, there is at least one $i$ with $\ve{e}_i \neq \tilde{\ve{e}}_i$ and thus $\mathcal{Y}_i' \neq \tilde{\mathcal{Y}}_i'$. This proves that each $\SL, \SE, \ve{E}$ yields a unique point in $B_{s,t,\epsilon}^\sub(\mathcal{S})$. Finally, there are $\binom{|\mathcal{Y}|}{s,t}$ possible solutions to choose the sets $\SL$ and $\SE$ and $(B_\epsilon^\sub(L)-1)^t$ error patterns $\ve{E}$.
\end{proof}
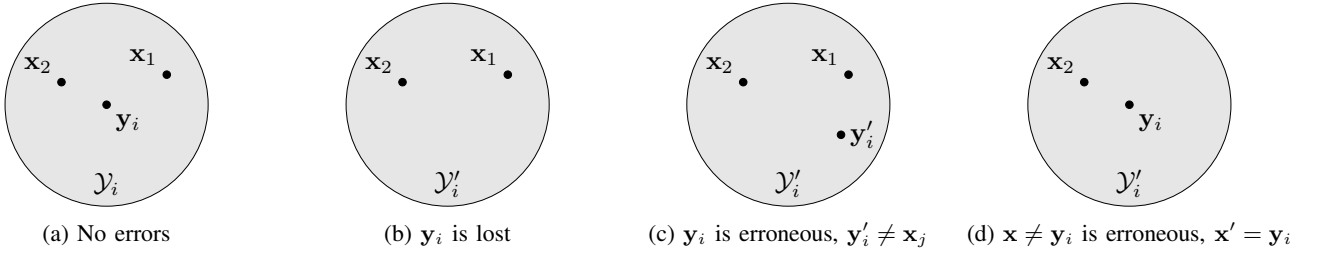
\begin{figure*}
	\begin{minipage}{.25\linewidth}
			\centering
			\begin{tikzpicture}[line width=0.125mm]
				\filldraw[fill=black!10!] (0,0) circle (1.35);
				\draw[fill] (0,0) circle (0.05) node [below right] {$\ve{y}_i$};
				\draw[fill] (0.8,0.4) circle (0.05) node [above left] {$\ve{x}_1$};
				\draw[fill] (-0.6,0.3) circle (0.05) node [above left] {$\ve{x}_2$};
				\draw (0.0,-1.05) node {$\mathcal{Y}_i$};
		\end{tikzpicture}
		\subcaption{No errors}
	\end{minipage}%
	\begin{minipage}{.25\linewidth}
		\centering
		\begin{tikzpicture}[line width=0.125mm]
			\filldraw[fill=black!10!] (0,0) circle (1.35);
			\draw[fill] (0.8,0.4) circle (0.05) node [above left] {$\ve{x}_1$};
			\draw[fill] (-0.6,0.3) circle (0.05) node [above left] {$\ve{x}_2$};
			\draw (0.0,-1.05) node {$\mathcal{Y}_i'$};
		\end{tikzpicture}
		\subcaption{$\ve{y}_i$ is lost}
	\end{minipage}%
	\begin{minipage}{.25\linewidth}
		\centering
		\begin{tikzpicture}[line width=0.125mm]
			\filldraw[fill=black!10!] (0,0) circle (1.35);
			\draw[fill] (0.7,-0.4) circle (0.05) node [right] {$\ve{y}_i'$};
			\draw[fill] (0.8,0.4) circle (0.05) node [above left] {$\ve{x}_1$};
			\draw[fill] (-0.6,0.3) circle (0.05) node [above left] {$\ve{x}_2$};
			\draw (0.0,-1.05) node {$\mathcal{Y}_i'$};
		\end{tikzpicture}
		\subcaption{$\ve{y}_i$ is erroneous, $\ve{y}_i' \neq \ve{x}_j$}
	\end{minipage}%
	\begin{minipage}{.25\linewidth}
		\centering
		\begin{tikzpicture}[line width=0.125mm]
			\filldraw[fill=black!10!] (0,0) circle (1.35);
			\draw[fill] (0,0) circle (0.05) node [below right] {$\ve{y}_i$};
			\draw[fill] (-0.6,0.3) circle (0.05) node [above left] {$\ve{x}_2$};
			\draw (0.0,-1.05) node {$\mathcal{Y}_i'$};
		\end{tikzpicture}
		\subcaption{$\ve{x} \neq \ve{y}_i$ is erroneous, $\ve{x}' = \ve{y}_i$}
	\end{minipage}%
	\caption{Cases for error patterns in Lemma \ref{lemma:substitutions:subset}}
	\label{fig:error:patterns}
\end{figure*}
This means, that if a set $\mathcal{S} \in \mathcal{X}_M^L$ contains an $\epsilon$-substitution-correcting code $\mathcal{Y}$ with cardinality $|\mathcal{Y}|$, the error ball has size at least $|B_{s,t,\epsilon}^\sub(\mathcal{S})| \geq \binom{|\mathcal{Y}|}{s,t}  B_{\epsilon}^\sub(L)^{t}$. Interestingly, for an appropriate choice of parameters, most of the sets $\mathcal{S} \in \mathcal{X}_M^L$ have the property of containing a large $\epsilon$-error-correcting code. To establish the fact, we need the following lemma.
\begin{lemma} \label{lemma:non:ecc:sets}
	Let $\mathcal{Y}\subseteq\mathcal{S}$ be the largest $\epsilon$-error-correcting code (error type $\ET$) with $B_\epsilon^\ET(\ve{y}_1) \cap B_\epsilon^\ET(\ve{y}_2) = \emptyset$ for all $\ve{y}_1, \ve{y}_2 \in \mathcal{Y}$ and $\ve{y}_1 \neq \ve{y}_2$. The number of sets $\mathcal{S} \subseteq \mathcal{X}_M^L$ with $|\mathcal{Y}| \leq K$, denoted as $D(K)$, is at most
	$$ D(K) \leq \binom{2^L}{K} \binom{KV^\ET_{\epsilon}}{M-K}, $$
	where
	$$V_\epsilon^\ET =  \underset{\ve{x} \in \Sigma_2^L}{\max} \, |\{ \ve{y} \in \Sigma_2^L: B_\epsilon^\ET(\ve{x}) \cap B_\epsilon^\ET(\ve{y}) \neq \emptyset \}|$$
	is the maximum number of sequences $\ve{y} \in \Sigma_2^L$ that have intersecting error balls $B_\epsilon^\ET(\cdot)$ with any $\ve{x} \in \Sigma_2^L$.
\end{lemma}
\begin{proof}
	Consider the following procedure on a set $\mathcal{S} \in \mathcal{X}_M^L$ whose largest $\epsilon$-error-correcting subset $\mathcal{Y} \subseteq \mathcal{S}$ has size at most $K$. Write $\mathcal{S}^{(1)} \triangleq \mathcal{S}$. Take an arbitrary word $\ve{x}^{(1)} \in \mathcal{S}^{(1)}$ and remove all words $\ve{y} \in \Sigma_2^L$ with intersecting error balls, i.e. $B_\epsilon^\ET(\ve{x}) \cap B_\epsilon^\ET(\ve{y}) \neq \emptyset$ from $\mathcal{S}^{(1)}$. Then select an arbitrary sequence from the resulting set $\mathcal{S}^{(2)}$, and, again, remove all elements with intersecting error balls. Continue this procedure until $S^{(j+1)} = \emptyset$. This procedure will stop after at most $j \leq K$ steps, since otherwise $\ve{x}_1, \dots, \ve{x}_{K+1}$ would form an $\epsilon$-error-correcting code. Hence, each such set ${\mathcal{S}}$ can be constructed by first selecting $K$ arbitrary, distinct words $\ve{x}_1,\dots,\ve{x}_{K}$ and then choosing the remaining $M-K$ words to have intersecting error balls with at least one of the $\ve{x}_1,\dots,\ve{x}_{K}$.
\end{proof}
While the bound from Lemma \ref{lemma:non:ecc:sets} may not seem particularly strong, it can be used to show that the number of sets that do not contain an $\epsilon$-substitution-correcting code of large size is negligible with respect to the sets that do contain an $\epsilon$-substitution-correcting code. We will elaborate this result and use it in the following to prove an upper bound on the size of $(s,t,\epsilon)_\sub$-correcting codes.
\begin{thm} \label{thm:substitutions:asymptotic} For fixed $s,t,\epsilon \in \mathbb{N}_0$ and $0<\beta<1$, any $(s,t,\epsilon)_\sub$-correcting code $\Code \subseteq \mathcal{X}_M^L$ satisfies
	$$ |\Code| \lesssim \frac{\binom{2^L}{M-s}}{\binom{M}{s,t}\binom{L}{\epsilon}^t}, $$
	when $M \rightarrow \infty$ with $M = 2^{\beta L}$. The redundancy is at least
	$$ r(\Code) \geq sL + t \log M + t\epsilon \log L - \log\left(s!t!\epsilon!^t\right) + o(1), $$
\end{thm}
\begin{proof}
	Denote by $\mathcal{D} \subseteq \mathcal{X}_M^L$, the set of all $\mathcal{S} \in \mathcal{X}_M^L$, which contain an $\epsilon$-substitution-correcting code $\mathcal{Y} \subseteq \mathcal{S}$ of size larger than $|\mathcal{Y}| > M-y(M)$, where we define $y(M) = M/\log M$. The remaining sets are comprised in $\mathcal{D}^\mathrm{C} = \mathcal{X}_M^L \setminus \mathcal{D}$. With the partition $\mathcal{D}\cup\mathcal{D}^\mathrm{C} = \mathcal{X}_M^L$, it follows that the cardinality of any $(s,t,\epsilon)_\sub$-correcting code $\Code \subseteq \mathcal{X}_M^L$ is at most
	\begin{align*}
		|\Code| = |\Code \cap \mathcal{D}| + |\Code \cap \mathcal{D}^\mathrm{C}|\leq \frac{\Big|\bigcup\limits_{\mathcal{S} \in \mathcal{D}}B_{s,t,\epsilon}^\sub(\mathcal{S})\Big|}{\underset{\mathcal{S} \in \mathcal{D}}{\min}|B_{s,t,\epsilon}^\sub(\mathcal{S})|} + |\mathcal{D}^\mathrm{C}|.
	\end{align*}
	The first term follows from a sphere packing bound on all sets $\mathcal{S} \in\mathcal{D}$. The numerator counts the total number of possible channel outputs and the denominator is a lower bound on the error ball size for all sets $\mathcal{S} \in \mathcal{D}$. Since each channel output is a set of sequences of size $M-s-t$ up to $M-s$, we have
	$$ \Big|\bigcup\limits_{\mathcal{S} \in \mathcal{D}}B_{s,t,\epsilon}^\sub(\mathcal{S})\Big| \leq \sum_{i=0}^{t} \binom{2^L}{M-s-i}. $$
	From Lemma \ref{lemma:substitutions:subset} it is known that
	$$ \underset{\mathcal{S} \in \mathcal{D}}{\min}|B_{s,t,\epsilon}^\sub(\mathcal{S})| \geq \binom{M-y(M)}{s,t}(B_\epsilon^\sub(L)-1)^t, $$
	and applying Lemma \ref{lemma:non:ecc:sets}, we find that $|\mathcal{D}^\mathrm{C}| \leq D(M-y(M))$. It follows that
	\begin{align*}
		|\Code| &\leq \frac{\sum_{i=0}^{t} \binom{2^L}{M-s-i}}{\binom{M-y(M)}{s,t}(B_\epsilon^\sub(L)-1)^t} + D(M-y(M)) \\
		&= \frac{\sum_{i=0}^{t} \binom{2^L}{M-s-i}}{\binom{M-y(M)}{s,t}(B_\epsilon^\sub(L)-1)^t} (1+\Delta),
	\end{align*}
	where $\Delta$ accounts for $D(M-y(M))$ and is defined implicitly as in the following equation. We will show that for our choice of $y(M)$, the first summand dominates the bound, i.e. $\Delta \rightarrow 0$ for $M \rightarrow \infty$. We obtain
	\begin{align*}
		\log \Delta =&\log  \frac{ D(M-y(M))\binom{M-y(M)}{s,t}(B_\epsilon^{\sub}(L)-1)^{t} }{\sum_{i=0}^{t} \binom{2^L}{M-s-i}} \\
		\overset{(a)}{\leq} & \log \frac{\binom{2^L}{M-y(M)}\binom{(M-y(M)) B_{2e}^\sub(L)}{y(M)}}{\binom{2^L}{M-s}} + O(L) \\
		\overset{(b)}{\leq} & - \frac{1-\beta}{\beta}M + o(M),
	\end{align*}
	where for inequality $(a)$ we used $|V^\sub_{\epsilon}(\ve{x})| = B_{2\epsilon}^\sub(L)$ for all $\ve{x} \in \Sigma_2^L$, $\log \binom{M-y(M)}{s,t} = O(L)$ and $t\log(B_\epsilon^{\sub}(L)-1) = O(\log L)$. Inequality $(b)$ follows from an application of Lemma~\ref{lemma:asymtptotic:good:sets} with $z(L) = 2^L/((M-y(M))B_{2\epsilon}^\sub(L))$. Therefore, $\Delta \rightarrow 0$, as $M \rightarrow \infty$ and $D(M-y(M))$ is asymptotically negligible. We obtain for any $(s,t,\epsilon)_\sub$-correcting code $\Code \subseteq \mathcal{X}_M^L$
	$$ |\Code| \lesssim \frac{\sum_{i=0}^{t} \binom{2^L}{M-s-i}}{\binom{M-y(M)}{s,t}(B_\epsilon^\sub(L)-1)^t} \sim \frac{\binom{2^L}{M-s}}{\binom{M}{s,t}\binom{L}{\epsilon}^t}. $$
	The redundancy is asymptotically at least
	\begin{align*}
		r(\Code) & = \log \frac{\binom{2^L}{M}}{|\Code|} \geq \log \frac{\binom{2^L}{M}\binom{M}{s} \binom{M-s}{t} \binom{L}{\epsilon}^{t}}{\binom{2^L}{M-s}}  + o(1) \\
		& \geq s \log (2^L-M) + t\log \left(ML^\epsilon\right) -\log (s!t!\epsilon!^t) + o(1)  \\ 
		& = sL + t \log M + t\epsilon \log L - \log\left(s!t!\epsilon!^t\right) + o(1),
	\end{align*}
	where we used that $\log \binom{M}{s} = s \log M - \log s! + o(1)$, $\log \binom{M-s}{t} = t \log M - \log t! + o(1)$ and $t\log \binom{L}{\epsilon} = t(\epsilon\log L - \log \epsilon!) + o(1)$.
\end{proof}
In particular, for $s=0$ and $\epsilon=1$, the redundancy of any $(0,t,1)_\sub$-correcting code $\Code \subseteq \mathcal{X}_M^L$ is at least $t \log (ML)-\log t!$ bits. Note that this coincides with the results from \cite{SRB18} for $t=1$. Comparing the bound on the redundancy stated in Theorem \ref{thm:substitutions:asymptotic} with the well known sphere packing bound for conventional $\epsilon$-substitution-correcting block codes, $\log B_\epsilon^{\sub}(L)$, yields an interesting interpretation of the $(0,t,1)_\sub$ channel. While it seems intuitive that the redundancy required is at least $t \log (ML)-\log t!$ bits, since there are $t$ errors inside a total of $ML$ symbols, it is interesting that from a sphere packing point of view, the fact the sequences are not ordered does appear to require as much redundancy as not knowing the distribution of the errors in an ordered array. While Theorem \ref{thm:substitutions:asymptotic} is formulated for a fixed number of errors $s,t$, we will find a bound for the case, when number of erroneous sequences $t$ is scaling with $M$ in the following.
\begin{thm} \label{thm:substitution:scaling}
	For fixed $s,\epsilon \in \mathbb{N}_0$ and fixed $0<\beta<1$, any $(s,M-s,\epsilon)_\sub$-correcting code $\Code \subseteq \mathcal{X}_M^L$ satisfies
	$$ r(\Code) \geq M \epsilon \log L + O(M), $$
	when $M \rightarrow \infty$ with $M = 2^{\beta L}$.
\end{thm}
\begin{proof}
	We follow a similar outline as in the proof for Theorem \ref{thm:substitutions:asymptotic}. Denote by $\mathcal{D} \subseteq \mathcal{X}_M^L$, the set of all $\mathcal{S} \in \mathcal{X}_M^L$, which contain an $\epsilon$-substitution-correcting code $\mathcal{Y} \subseteq \mathcal{S}$ of size $|\mathcal{Y}| > M-y(M)$, where we define $y(M) = M/\log \log M$ and $\mathcal{D}^\mathrm{C} = \mathcal{X}_M^L \setminus \mathcal{D}$. Allowing only $t=M-s-y(M)$ erroneous sequences, we can apply Lemma \ref{lemma:substitutions:subset} and obtain
	$$ |B_{s,t,\epsilon}^\sub(\mathcal{S})| \geq \binom{M-y(M)}{s}(B_\epsilon^\sub(L)-1)^{M-y(M)-s}, $$
	for all $\mathcal{S} \in \mathcal{D}$. It follows that
	$$ |\Code| \leq \frac{\sum_{i=0}^{M-y(M)-s
		} \binom{2^L}{M-s-i}}{\binom{M-y(M)}{s}(B_\epsilon^\sub(L)-1)^{M-y(M)-s}} (1+\Delta). $$
	We will show that $\Delta \rightarrow 0$ for $M \rightarrow \infty$. We obtain
	\begin{align*}
	\log \Delta =&\log  \frac{\binom{M-y(M)}{s}(B_\epsilon^{\sub}(L)-1)^{M-y(M)-s} D(M-y(M))}{\sum_{i=0}^{M-y(M)-s} \binom{2^L}{M-s-i}} \\ \overset{(a)}{\leq} & \log \frac{\binom{2^L}{M-y(M)}\binom{(M-y(M)) B_{2\epsilon}^\sub(L)}{y(M)}}{\binom{2^L}{M-s}} + M\epsilon\log L + O(L) \\
	\overset{(b)}{\leq} & -\frac{ML}{\log \log M} + M \epsilon \log L + o\left(\frac{M}{\log \log M}\right) \\
	= & - \frac{ML}{\log(\beta L)} + O(M \log L)
	\end{align*}
	where in inequality $(a)$ we used $\log \binom{M-y(M)}{s} = O(L)$. For inequality $(b)$ we applied Lemma \ref{lemma:asymtptotic:good:sets} with $z(L) = 2^L/((M-y(M))B_{2\epsilon}^\sub(L))$. Therefore, $\Delta \rightarrow 0$, as $M \rightarrow \infty$. We obtain for any $(s,M-s,\epsilon)_\sub$-correcting code $\Code \subseteq \mathcal{X}_M^L$
	$$ |\Code| \lesssim \frac{\sum_{i=0}^{M-y(M)-s} \binom{2^L}{M-s-i}}{\binom{M-y(M)}{s}(B_\epsilon^\sub(L)-1)^{M-y(M)-s}} \lesssim \frac{\binom{2^L}{M-s}}{\binom{M}{s}\binom{L}{\epsilon}^{M-y(M)-s}}. $$
	Therefore, the redundancy satisfies
	\begin{align*}
	r(\Code) & = \log \frac{\binom{2^L}{M}}{|\Code|} \geq \log \frac{\binom{2^L}{M}\binom{M}{s}  \binom{L}{\epsilon}^{M-y(M)-s}}{\binom{2^L}{M-s}}  + o(1) \\
	& \geq sL + (M-y(M)-s)\epsilon\log (L/\epsilon) -\log s! + o(1).
	\end{align*}%
\end{proof}
\subsection{Asymptotic Bounds for Deletion Errors} \label{ss:asmyptotic:boubds:deletions}
We will now turn to derive an asymptotic bound on the cardinality of $(s,t,\epsilon)_\del$-correcting codes. Note that it is possible to use the technique that we present here also for insertion errors, however this is deferred to future work. Since the deletion ball is non-uniform, it is not directly possible to use an analogue of Lemma \ref{lemma:substitutions:subset} as in Theorem \ref{thm:substitutions:asymptotic}. We will therefore slightly adapt our arguments and use the fact that, although the deletion ball size is non-uniform, most of the deletion balls have a similar size. It has been shown in \cite{Lev66} that
$$ |S_\epsilon^\del(\ve{x})| \geq \binom{||\ve{x}||-\epsilon+1}{\epsilon} $$
and most words $\ve{x} \in\Sigma_2^n$ have roughly $L/2$ runs. We will elaborate this result in the following.
\begin{lemma} \label{lemma:number:words:run}
	Let $\rho \in \mathbb{N}$. The number of words with less than $L/2-\rho$ runs satisfies
	$$ \left|\left\{ \ve{x} \in \Sigma_2^L: ||\ve{x}|| < \frac{L}{2} - \rho \right\}\right| \leq \frac{2^L}{\e^{\frac{2\rho^2}{L}}}. $$
\end{lemma}
\begin{proof}
	The number of words $\ve{x}\in \Sigma_2^L$ with exactly $i$ runs, i.e., $||\ve{x}|| = i$ is given by $2\binom{L-1}{i-1}$. Therefore, the number of words with less than $L/2 - \rho$ runs is given by
	\begin{align*}
	|\{\ve{x} & \in \Sigma_2^L :||\ve{x}|| < L/2-\rho \}|  = 2\sum_{i=1}^{L/2-\rho-1} \binom{L-1}{i-1} \\ & \overset{(a)}{\leq} \sum_{i=1}^{L/2-\rho} \binom{L}{i} \overset{(b)}{\leq} \frac{2^L}{\e^{\frac{2 \rho^2}{L}}},
	\end{align*}
	where we used $\binom{L-1}{i-1} \leq \frac12\binom{L}{i}$ for $i\leq\frac{L}{2}$ in inequality $(a)$ and Hoeffding's inequality \cite{Hoe63} on the binomial sum in $(b)$.
\end{proof}
Next, we find a lower bound on the ball size $B_{s,t,\epsilon}^\del(\mathcal{S})$, for sets, which contain a deletion-correcting code.
\begin{lemma} \label{lemma:lower:bound:deletion}
	Let $\mathcal{Y} \subseteq \mathcal{S} \in \mathcal{X}_M^L$ be an $\epsilon$-deletion-correcting code, i.e. $B_\epsilon^\del(\ve{y}_1) \cap B_\epsilon^\del(\ve{y}_2) = \emptyset$ for all $\ve{y}_1, \ve{y}_2 \in \mathcal{Y}$ and $\ve{y}_1 \neq \ve{y}_2$. Further, let $s+t\leq |\mathcal{Y}|$. Then,
	$$|B_{s,t,\epsilon}^\sub(\mathcal{S})| \geq \sum_{\substack{\SE,\SL \subseteq \mathcal{Y}, \SE\cap\SL=\emptyset,\\ |\SL| = s, |\SE| = t}} \,\, \prod_{\ve{y} \in \SE} |S_{\epsilon}^\del(\ve{y})|,$$
\end{lemma}
\begin{proof}
	We will find a lower bound on the number of words inside the error ball $|B^\del_{s,t,\epsilon}(\mathcal{S})|$ by counting distinct elements $S' \in B^\del_{s,t,\epsilon}(\mathcal{S})$ in the following way. Choose two arbitrary distinct sets $\SL,\SE=\{\ve{y}_{e_1}, \dots, \ve{y}_{e_t} \} \subseteq \mathcal{Y}$ with $|\SL| = s$ and $|\SE| = t$ and choose a set of erroneous outcomes $\SE' = \{\ve{y}_1', \dots, \ve{y}_t' \}$, where $\ve{y}_i' \in S_\epsilon^\del(\ve{y}_{e_i})$. Note that we delete exactly $\epsilon$ symbols from each $\ve{y}_{e_i}$ and thus $\ve{y}_i' \in \Sigma_2^{L-\epsilon}$. Denote by $\SL, \SE, \SE'$ and $\tilde{\SL}, \tilde{\SE}, \tilde{\SE}'$ two different choices of error realizations and let $\mathcal{S}'$ and $\tilde{\mathcal{S}}'$ be the corresponding received sets. If $\SL\cup \SE \neq\tilde{\SL}\cup \tilde{\SE} $, then $\mathcal{S}' \neq \tilde{\mathcal{S}}'$, as the resulting error-free sequences in $\mathcal{S}'$ and $\tilde{\mathcal{S}}$ of length $L$ are different. In the case $\SL\cup \SE =\tilde{\SL}\cup \tilde{\SE}$ and $\SE \neq \tilde{\SE}$, it follows that $\SE' \neq \tilde{\SE}'$, as the erroneous outcomes are chosen out of the $\epsilon$ deletion spheres from an $\epsilon$-deletion-correcting code. Finally, if $\SL\cup \SE =\tilde{\SL}\cup \tilde{\SE}$ and $\SE = \tilde{\SE}$ it follows that $\SL = \tilde{\SL}$ and thus $\SE' \neq \tilde{\SE}'$ as we chose $\SL, \SE, \SE'$ and $\tilde{\SL}, \tilde{\SE}, \tilde{\SE}'$ to be different. Hence, for each choice of error realizations $\SL, \SE, \SE'$, we obtain a unique element in $B_{s,t,\epsilon}^\del(\mathcal{S})$. Counting the number of choices yields the lemma.
\end{proof}
This allows to formulate the following theorem.
\begin{thm} \label{thm:deletion:asymptotic} For fixed $s,t,\epsilon \in \mathbb{N}_0$ and $0<\beta<1$, any $(s,t,\epsilon)_\del$-correcting code $\Code \subseteq \mathcal{X}_M^L$ satisfies
	$$ |\Code| \lesssim \frac{\binom{2^L}{M-s-t}\binom{2^{L-\epsilon}}{t}}{\binom{M}{s,t}\binom{L/2}{\epsilon}^t} $$
	when $M \rightarrow \infty$ with $M = 2^{\beta L}$. The redundancy is at least
	$$ r(\Code) \geq sL+t\epsilon \log L -\log (s!\epsilon!^t)  + o(1). $$
\end{thm}
\begin{proof}
	Denote by $\mathcal{D}_\mathrm{r} \subseteq \mathcal{X}_M^L$, the set of all $\mathcal{S} \in \mathcal{X}_M^L$, which contain more than $M-y(M)$ sequences with $||\ve{x}|| \geq L/2-\rho(L)$, where we choose $y(M) = M/\log M$ and $\rho(L) = \sqrt{L \ln L}$. Further, let $\mathcal{D}_\mathrm{e} \subseteq \mathcal{X}_M^L$ be all sets $\mathcal{S} \in \mathcal{X}_M^L$ that contain an $\epsilon$-deletion-correcting code $\mathcal{Y}\subseteq \mathcal{S}$ of size $|\mathcal{Y}| >M-y(M)$ and let $\mathcal{D} = \mathcal{D}_\mathrm{r} \cap \mathcal{D}_\mathrm{e}$. The remaining sets are comprised in $\mathcal{D}^\mathrm{C} = \mathcal{X}_M^L \setminus \mathcal{D}$. Since $\mathcal{D}$ and $\mathcal{D}^\mathrm{C}$ are a partition of $\mathcal{X}_M^L$, every $(s,t,\epsilon)_\del$-correcting code $\Code\subseteq \mathcal{X}_M^L$ satisfies
	$$ |\Code| = |\Code \cap \mathcal{D}| + |\Code \cap \mathcal{D}^\mathrm{C}|\leq \frac{\Big|\bigcup\limits_{\mathcal{S} \in \mathcal{D}}B_{s,t,\epsilon}^\del(\mathcal{S})\Big|}{\underset{\mathcal{S} \in \mathcal{D}}{\min}|B_{s,t,\epsilon}^\del(\mathcal{S})|} + |\mathcal{D}^\mathrm{C}|. $$
	The number of received sets after a loss of exactly $s$ sequences and $t$ sequences with exactly $\epsilon$ deletions each is at most
	$$ \Big|\bigcup\limits_{\mathcal{S} \in \mathcal{D}}B_{s,t,\epsilon}^\del(\mathcal{S})\Big| \leq \binom{2^L}{M-s-t} \binom{2^{L-\epsilon}}{t}, $$
	as each received set consists of $M-s-t$ error-free sequences and $t$ sequences of length $L-\epsilon$. Each $\mathcal{S} \in \mathcal{D}$ contains less than $y(M)$ sequences, which do not belong to the $\epsilon$-deletion-correcting code $\mathcal{Y}$ and less than $y(M)$ (possibly different) sequences with $||\ve{x}|| < L/2-\rho(L)$. Thus, at least $M-2y(M)$ sequences form an $\epsilon$-deletion-correcting code and satisfy $||\ve{x}|| \geq  L/2-\rho(L)$ and by Lemma \ref{lemma:lower:bound:deletion}, we have
	$$|B_{s,t,\epsilon}^\del(\mathcal{S})| \geq \binom{M-2y(M)}{s,t} \binom{L/2-\rho(L)-\epsilon}{\epsilon}^{t} $$
	for each $\mathcal{S} \in \mathcal{D}$. The number of remaining sets $\mathcal{S} \notin \mathcal{D}$ satisfies $|\mathcal{D}^\mathrm{C}| = |\mathcal{X}_M^L \setminus \mathcal{D}| \leq |\mathcal{X}_M^L \setminus \mathcal{D}_\mathrm{r}| + |\mathcal{X}_M^L\setminus \mathcal{D}_\mathrm{e}|$. Each such set contains at least $y(M)$ sequences with $||\ve{x}|| < L/2-\rho(L)$ or does not contain an $\epsilon$-deletion-correcting code of size more than $M-y(M)$. By Lemma \ref{lemma:number:words:run}, we have that
	$$ |\mathcal{X}_M^L \setminus \mathcal{D}_\mathrm{r}| \leq \binom{2^L}{M-y(M)}\binom{2^L/L^2}{y(M)}, $$
	for large enough $L$, as each $\mathcal{S} \in \mathcal{X}_M^L \setminus \mathcal{D}_\mathrm{r}$ can be constructed by choosing $y(M)$ sequences to have less than $L/2-\rho(L)$ runs and the remaining sequences are chosen arbitrarily. Further, using Lemma \ref{lemma:non:ecc:sets}, it follows that
	$$ |\mathcal{X}_M^L\setminus \mathcal{D}_\mathrm{e}| \leq \binom{2^L}{M-y(M)} \binom{K V_\epsilon^\del }{y(M)}, $$
	where $V_\epsilon^\del = \max_{\ve{x} \in \Sigma_2^L} \, |\{ \ve{y} \in \Sigma_2^L: B_\epsilon^\del(\ve{x}) \cap B_\epsilon^\del(\ve{y}) \neq \emptyset \}|$. This number can be bounded from above by the following consideration. Given $\ve{x} \in \Sigma_2^L$, each ${\ve{y}} \in \Sigma_2^L$ can be constructed by first deleting $\epsilon$ symbols from $\ve{x}$ and then inserting $\epsilon$ arbitrary symbols to the result. Using $|S_\epsilon^\del(\ve{x})| \leq \binom{L}{\epsilon}$ for all $\ve{x} \in \Sigma_2^L$ and $|S_\epsilon^\ins(\ve{x}')| = \sum_{i=0}^{\epsilon} \binom{L}{i}= B_\epsilon^\sub$ for all $\ve{x}' \in \Sigma_2^{L-\epsilon}$ yields $V_\epsilon^\del \leq \binom{L}{\epsilon} B_\epsilon^\sub$. It follows that the size of any $(s,t,\epsilon)_\del$-correcting code $\Code \subseteq \mathcal{X}_M^L$ is at most
	\begin{align*}
	|\Code| \leq & \frac{\binom{2^L}{M-s-t}\binom{2^{L-\epsilon}}{t}}{\binom{M-2y(M)}{s,t} \binom{L/2-\rho(L)-\epsilon}{\epsilon}^{t}} \\
	& + \binom{2^L}{M-y(M)} \left(\binom{2^L/L^2}{y(M)}+ \binom{K \binom{L}{\epsilon} B_\epsilon^\sub}{y(M)}\right) \\
	=& \frac{\binom{2^L}{M-s-t}\binom{2^{L-\epsilon}}{t}}{\binom{M-2y(M)}{s,t} \binom{L/2-\rho(L)-\epsilon}{\epsilon}^{t}} (1+\Delta_\mathrm{r}+\Delta_\mathrm{e}).
	\end{align*}
	We will show now that $\Delta_\mathrm{r} \rightarrow 0$ and $\Delta_\mathrm{e} \rightarrow 0$ for $M \rightarrow \infty$.
	\begin{align*}
	\log \Delta_\mathrm{r} =&\log \frac{\binom{M}{s,t} \binom{\frac{L}{2}}{\epsilon}^{t} \binom{2^L/L^2}{y(M)}\binom{2^L}{M-y(M)} }{\binom{2^L}{M-s-t}\binom{2^{L-\epsilon}}{t}} + o(1) \\
	=& \log \frac{\binom{2^{L}/L^2}{y(M)}  \binom{2^L}{M-y(M)} }{\binom{2^L}{M-s-t}} + O(L) \\
	\overset{(a)}{\leq} & -\frac{M}{\log M} \log \log M + O\left( \frac{M}{\log M}\right),
	\end{align*}
	where we applied Lemma \ref{lemma:asymtptotic:good:sets} in inequality $(a)$. Hence, $\Delta_\mathrm{r} \rightarrow 0$ for $M \rightarrow \infty$. Analogous to the proof of Theorem \ref{thm:substitutions:asymptotic}, it can be shown that $\Delta_\mathrm{e} \rightarrow 0$ for $M \rightarrow \infty$. We obtain for the maximum size of a $(s,t,\epsilon)_\del$-correcting code
	$$ |\Code| \lesssim \frac{\binom{2^L}{M-s-t}\binom{2^{L-\epsilon}}{t}}{\binom{M}{s,t}\binom{L/2}{\epsilon}^t}. $$
	The redundancy is consequently at least
	\begin{align*}
	r(\Code) &= \log \binom{2^L}{M} - \log |\Code| \geq \log \frac{\binom{2^L}{M}\binom{M}{s,t}\binom{L/2}{\epsilon}^t}{\binom{2^L}{M-s-t}\binom{2^{L-\epsilon}}{t}}  +  o(1) \\
	&= sL+t\epsilon \log L -\log (s!\epsilon!^t)  + o(1).
	\end{align*}
\end{proof}
The result of Theorem \ref{thm:deletion:asymptotic} is particularly interesting, when comparing with Theorem \ref{thm:substitutions:asymptotic}, which depicts the case of substitution errors inside the sequences. It can be seen that correcting substitutions requires $t \log M - \log t!$ more bits of redundancy as compared to insertion or deletion errors only. While this seems surprising, there is a practical reason for this phenomena. For the case of insertion or deletion errors, it is directly possible to identify erroneous sequences, by checking their length to be different from $L$. This is not possible for substitution errors, and erroneous sequences can be confused with correct sequences, which means that additional redundancy is required for detecting the erroneous sequences. In fact, we will show in Construction \ref{con:single:insertion}, how to constructively exploit the identification of erroneous sequences for the case of $(0,1,1)_\mathbb{D}$ deletion errors and obtain a code that asymptotically achieves the bound from Theorem \ref{thm:deletion:asymptotic}. In the following we derive a sphere packing bound for the case when the number of erroneous sequences scales with $M$.
%
\vspace{-.01cm}
\begin{thm} \label{thm:deletion:scaling}
	For fixed $s,\epsilon \in \mathbb{N}_0$ and fixed $0<\beta<1$, any $(s,M-s,\epsilon)_\del$-correcting code $\Code \subseteq \mathcal{X}_M^L$ satisfies
	$$ r(\Code) \geq M \epsilon \log L +O(M), $$
	when $M \rightarrow \infty$ with $M = 2^{\beta L}$.
\end{thm}
\begin{proof}
	The proof is similar to that of Theorem \ref{thm:deletion:asymptotic} and we use the same notation for $\mathcal{D} = \mathcal{D}_\mathrm{r} \cap \mathcal{D}_\mathrm{e}$ for sets that contain an $\epsilon$-deletion-correcting code of size $|\mathcal{Y}| > M-y(M)$ and more than $M-y(M)$ sequences with at least $||\ve{x}|| \geq L/2 -\rho(L)$ runs, where $y(M) =  M/\log \log M$ and $\rho(L) = \sqrt{L/2 \ln L \log^2 L}$. With Lemma \ref{lemma:lower:bound:deletion}, it follows
	$$ |B_{s,t,\epsilon}^\del(\mathcal{S})| \geq \binom{M-2y(M)}{s} \binom{L/2-\rho(L)-\epsilon}{\epsilon}^{M-2y(M)-s} $$
	for all $\mathcal{S} \in \mathcal{D}$. It follows that the size of any $(s,t,\epsilon)_\del$-correcting code $\Code \subseteq \mathcal{X}_M^L$ is at most
	\begin{align*}
	|\Code| \leq \frac{\binom{2^L}{2y(M)}\binom{2^{L-\epsilon}}{M-2y(M)-s}}{\binom{M-2y(M)}{s} \binom{L/2-\rho(L)-\epsilon}{\epsilon}^{M-y(M)-s}} (1+\Delta_\mathrm{r}+\Delta_\mathrm{e}).
	\end{align*}
	We will show now that $\Delta_\mathrm{r} \rightarrow 0$ and $\Delta_\mathrm{e} \rightarrow 0$ for $M \rightarrow \infty$.
	\begin{align*}
	\log \Delta_\mathrm{r} \leq &\frac{ \binom{L/2-\rho(L)-\epsilon}{\epsilon}^{M} \binom{2^L/L^{\log^2L}}{y(M)}\binom{2^L}{M-y(M)} }{\binom{2^L}{2y(M)}\binom{2^{L-\epsilon}}{M-2y(M)-s}} + O(L) \\
	\leq &\log\frac{\binom{2^L/L^{\log^2L}}{y(M)}  \binom{2^L}{M-y(M)} }{\binom{2^L}{2y(M)}\binom{2^{L-\epsilon}}{M-2y(M)-s}}
	 +M\epsilon \log L + O(L) \\
	\overset{(a)}{\leq} &\log \frac{\binom{2^L/L^{\log^2L}}{y(M)}  \binom{2^L}{M-y(M)} }{\binom{2^L}{M-s}} + M\epsilon \log L + O(M) \\
	\overset{(b)}{\leq} & -\frac{M \log^3L}{\log (\beta L)} + O(M\log L)
	\end{align*}
	where for inequality $(a)$ we used that
	$$\log \frac{\binom{2^L}{M-s}}{\binom{2^L}{2y(M)}\binom{2^{L-\epsilon}}{M-2y(M)-s}} \leq O(M)$$
	and applied Lemma \ref{lemma:asymtptotic:good:sets} in inequality $(b)$. Hence, $\Delta_\mathrm{r} \rightarrow 0$ for $M \rightarrow \infty$. Analogous to the proof of Theorem \ref{thm:substitution:scaling}, it can be shown that $\Delta_\mathrm{e} \rightarrow 0$ for $M \rightarrow \infty$. We obtain for the maximum size of a $(s,M-s,\epsilon)_\del$-correcting code
	$$ |\Code| \lesssim \frac{\binom{2^L}{2y(M)}\binom{2^{L-\epsilon}}{M-2y(M)-s}}{\binom{M}{s} \binom{L/2-\rho(L)-\epsilon}{\epsilon}^{M-y(M)-s}}.$$
	The redundancy is consequently at least
	\begin{align*}
	r(\Code) & \geq \log \frac{\binom{2^L}{M}\binom{M}{s} \binom{L/2-\rho(L)-\epsilon}{\epsilon}^{M-y(M)-s}}{\binom{2^L}{2y(M)}\binom{2^{L-\epsilon}}{M-2y(M)-s}} + o(1) \\ &\geq  M \epsilon \log L + O(M).
	\end{align*}
\end{proof}
\section{Code Constructions}\label{sec:const}
Having available suitable bounds on the redundancy of $(s,t,\epsilon)_\ET$-correcting codes, we now present several code constructions for DNA storage systems that are suitable for different types of errors $\ET$ and choices of parameters $s,t$ and $\epsilon$. Note that constructing codes for the considered channel model is surprisingly challenging. This can be explained as follows. In order to find efficient codes, one has to deal with both, the fact that sequences are received in an unordered fashion and also that sequences are distorted by random errors. Especially for the case when only few sequences contain errors, concatenated schemes, where each strand is protected by an inner code and an outer code is used over all strands, are suboptimal. This is because the inner code is "wasted" for the correct sequences, as they do not contain any errors and no inner code would be needed to obtain the correct strand. Therefore, it is important to design codes over the whole set of strands that use redundancy over different sequences and allow to correct errors in a subset of strands. We start with constructions that are suitable for an arbitrary number of errors per sequence and will elaborate more specialized constructions towards the end of this section.

\subsection{Indexing Sequences}
A common efficient way to combat the loss of ordering of sequences is to prepend an index to each sequence, which contains the position $i$ of the sequence. This approach has been discussed in different settings, e.g. \cite{KT18,HRT17}. The set of all possible sets of sequences with indexing is given by
$$ \Code_\mathrm{I}(M,L)  = \{ \mathcal{S} \in \mathcal{X}_M^L : \ve{x}_i = (\ve{I}(i),\ve{u}_i), i \in \{1,2,\dots, M \},$$
where $\mathbf{I}(i) \in \Sigma_2^{\lceil \log M \rceil}$ denotes the binary representation of $i-1$ and $\mathbf{u}_i \in \Sigma_2^{L-\lceil \log M \rceil}$ are arbitrary information vectors. Note that by this definition the prefix of an indexed sequence contains the index of a sequence $\mathrm{pref}_{\lceil \log M\rceil}(\ve{x}_i) = \ve{I}(i)$.

This requires an index $\mathbf{I}(i)$ of $\lceil \log M\rceil$ bits in each sequence so the maximum number of information bits that can be stored this way is $M(L-\lceil \log M\rceil)$ without any error correction. While this solution is attractive for its simplicity, it introduces already a redundancy, which increases linearly in $M$, which is stated in the following theorem.
\begin{thm} \label{thm:indexing:redundancy}
	For fixed $0<\beta<1$, the redundancy required for indexing sequences is given by
	$$r(\Code_\mathrm{I}(M,L)) = M(\lceil \log M \rceil - \log M+\log \e) + o(M),$$
	when $M \rightarrow \infty$ with $M = 2^{\beta L}$.
\end{thm}
\begin{proof}
	From $M = 2^{\beta L}$ with $0<\beta<1$, we have that $M = o(2^L)$ and $M = \omega(1)$, when $M \rightarrow \infty$. Therefore,
	\begin{align*}
	r(\Code_\mathrm{I}(M,L)) =& \log\binom{2^L}{M} - M(L-\lceil \log M\rceil)\\
	=& M(\left\lceil \log M  \right\rceil - \log M + \log \e) + o(M),
	\end{align*}
	where we used Lemma \ref{lemma:approx:binom}, which is derived in Appendix \ref{app:auxiliary:lemmas}, to characterize the binomial coefficient.
\end{proof}
This means that every construction which uses indexing already incurs a redundancy of at best roughly $M\log \e$ bits. Note that this amount can be significant, as the number of sequences $M$ is significantly larger than their length $L$, as explained in Remark \ref{rem:order:L:M}. However, in terms of code rate, it has been shown in \cite{HRT17} that for the case of no errors inside the sequences, the indexing approach is capacity achieving.

The following function, which collects all indices of a set of sequences will be useful for our constructions that are based on indexing.
\begin{defn}
	For any set $\mathcal{A}\subseteq \mathcal{X}_M^L$ we define
	$$\mathcal{I}(\mathcal{A}) = \bigcup_{\ve{x} \in \mathcal{A}} \{\mathrm{pref}_{\lceil \log M \rceil}(\ve{x})\}$$
	to be the set of indices of the sequences in $\mathcal{A}$. 
\end{defn}
Note that it is possible that $|\mathcal{I}(\mathcal{A})| < |\mathcal{A}|$, if one (or more) of the indices appear multiple times because of errors. 

\subsection{An Index-Based Construction using MDS Codes} \label{ss:index:based:construction:MDS}
The following construction is based on adding an index in front of all sequences $\ve{x}_i$ and using an MDS code over the $M$ sequences for error correction. For all $n$ and $k$, where $k\leq n$ we denote by $\mathsf{MDS}[n,k]$ an MDS code over any field of size at least $n-1$.

In Construction~\ref{con:index:rs}, the sequences $\ve{x}_i = (\ve{I}(i), \ve{u}_i)$ of each codeword set are constructed by writing a binary representation of the index, $\ve{I}(i)$, of length $\lceil \log M \rceil$ in the first part of each sequence. Then, the remaining part $\ve{u}_i$ is viewed as a symbol over the extension field $\mathbb{F}_{2^{L-\lceil \log M \rceil}}$, and $(\ve{u}_1,\ldots,\ve{u}_M)$  will form a codeword in some MDS code\footnote{Note that we assume $M \leq \sqrt{2^L}$ in this section to guarantee the existence of the MDS code \cite[ch. 11]{Rot06}. However, the case $M > \sqrt{2^L}$ can always be used by employing non-MDS codes.}. A similar construction has been used in \cite{HRT17}, where index-based constructions are analyzed for the correction of only a loss of sequences.
\begin{construction} \label{con:index:rs}
For all $M, L$, and a positive integer $\delta$, let $\Code_1 (M,L,\delta)$ be the code defined by \vspace{-0.5ex}
	\begin{align*}
	\Code_1 (M,L,\delta)= \{ \mathcal{S} \in \mathcal{X}_M^L : \, & \ve{x}_i = (\ve{I}(i), \ve{u}_i),  \\ 
	&(\ve{u}_1, \ldots, \ve{u}_M) \in \mathsf{MDS}[M,M-\delta] \}.
	\end{align*}
\end{construction}
This code provides a direct construction to correct a loss of sequences and erroneous sequences with an arbitrary amount of errors each. The error correction capability for several types of errors is summarized in the following lemma.
\begin{lemma} \label{lemma:index:rs:cor}
For all $M,L,\delta$, the code $\Code_1 (M,L,\delta)$ is
\begin{itemize}
	\item $(s,t,\bullet)_\mathbb{L}$-correcting for all $s+2t \leq \delta$,
	\item $(s,t,\bullet)_\ins$-correcting for all $s+t \leq \delta$,
	\item $(s,t,\bullet)_\del$-correcting for all $s+t \leq \delta$.
\end{itemize}
\end{lemma}
\begin{proof}
Denote by $\mathcal{S}'$ the received set after a loss of sequences and errors. We start with proving the lemma for the case of arbitrary edit errors. According to Definition \ref{def:dna:channel:error:ball}, we write $\SG, \SL,\SE$ as the sets of error-free, lost, and erroneous sequences, and $\SE'$ are the erroneous outcomes of the sequences in $\SE$. First, we observe that if we can recover the MDS codeword $\ve{U} = (\ve{u}_1, \ve{u}_2, \dots, \ve{u}_M)$, we can also recover $\mathcal{S}$ by prepending the index $\ve{I}(i)$ in front of each $\ve{u}_i$. Given $\mathcal{S}'$, we then create the received estimate word $\ve{U}'$ by declaring all positions $i$ with
	$$ |\{ \ve{x}' \in \mathcal{S}': \mathrm{pref}_{\lceil \log M \rceil} (\ve{x}') = \ve{I}(i) \}| \neq 1, $$
	i.e., for which there is not exactly one index in $\mathcal{S}'$, as erasures. The remaining positions in $\ve{U}'$ are filled with the corresponding symbols $\ve{u}'_i$. We will show that the number of erasures $s'$ and the number of errors $t'$ in $\ve{U}'$ satisfy $s' + 2t' \leq \delta$ by the following consideration. Consider a genie, which first only adds the error-free sequences $\SG$ to $\ve{U}'$. Since $|\SG|\geq M-s-t$ at least $M-s-t$ positions in $\ve{U}'$ have been filled with correct symbols, and thus there remain at most $s'\leq s+t$ erasures and $t'=0$ errors inside $\ve{U}'$ up to this point. Now, the genie successively adds the $t$ erroneous sequences $\SE'$ to $\ve{U}'$. Each of the $t$ erroneous sequences $\ve{x}' \in \SE'$ can have an arbitrary index $\mathrm{pref}_{\lceil \log M \rceil}(\ve{x}')$ from $0$ to $M-1$. If the erroneous sequence $\ve{x}'$ has an index of a position which is not occupied yet, $t'$ increases by one and $s'$ decreases by one, as this position is not declared as an erasure anymore and contains now an erroneous symbol. If the erroneous sequence $\ve{x}'$ has an index, which is already present, this position is declared as an erasure as explained above. Consequently, $s'$ increases by one for this case. Hence, the number of erasures in $\ve{U}'$ is bounded from above by $s' \leq s+t-t' + |\mathcal{I}(\SE') \cap \mathcal{I}(\SG)|$ where $|\mathcal{I}(\SE') \cap \mathcal{I}(\SG)|$ accounts for the situation when an erroneous sequence has the same index as an error-free one. The number of errors is at most $t' \leq |\mathcal{I}(\SE') \cap (\mathcal{I}(\SE) \cup \mathcal{I}(\SL))|$. Hence, $s'+2t' \leq s+t+t' + |\mathcal{I}(\SE') \cap \mathcal{I}(\SG)| \leq s+2t \leq \delta$, which proves the error correcting capability. 
	
	For the case of only insertion ($\ins$) and only deletion ($\del$) errors, it is possible to identify the erroneous sequences by checking their length to be larger (respectively smaller) than $L$. If these sequences are discarded, there are in total $s+t$ erasures inside the MDS codeword, which can be corrected, if $s+t\leq \delta$.
\end{proof}
Note that for the practically important case of a loss of sequences and combinations of substitution and deletion errors, $\Code_1(M,L,\delta)$ can correct all errors, if $s+2t_\sub + t_\del\leq \delta$, where $t_\sub$ is the number of sequences suffering from substitution errors only and $t_\del$ is the number of sequences with deletion errors. The same also holds for combinations of substitution and insertion errors. However, this is not true for combinations of substitutions, insertions \emph{and} deletions as a sequence that contains insertions and deletions might have length exactly $L$ and therefore cannot be erased. In this case, as elaborated in the proof, $s+2t \leq \delta$ has to hold. More generally, erroneous sequences which have length exactly $L$ require $2$ redundancy symbols inside the MDS codeword to be correctable, while sequences which have a different length only require a single symbol, as they can be detected as erroneous.

The redundancy of Construction~\ref{con:index:rs} is stated in the following theorem.
\begin{thm} \label{thm:red:index:rs}
For all $M,L,\delta$, the redundancy of the code $\Code_1 (M,L,\delta)$ is
	$$ r(\Code_1 (M,L,\delta)) = r(\Code_\mathrm{I}(M,L)) + \delta (L-\lceil \log M \rceil).$$
\end{thm}
\begin{proof}
	First, indexing the sequences requires a redundancy of $r(\Code_\mathrm{I}(M,L))$, which is derived in Theorem \ref{thm:indexing:redundancy}. Second, the MDS code has $\delta$ redundant symbols and thus there are $\delta(L-\lceil \log M \rceil)$ additional redundancy bits.
\end{proof}
While the redundancy of Construction \ref{con:index:rs} can be very large, especially for the case $M \gg L$, it provides some very useful features. First, it is possible to efficiently encode and decode this code using standard encoders and decoders for MDS codes. Second, it is not necessary to design the code for a specific number of errors $s$ and $t$, but rather their sum $s+2t$, which allows for a flexible decoding procedure.

\subsection{A Construction Based On Constant Weight Codes}

Imposing an ordering (e.g., lexicographic) onto the sequences in $\Sigma_2^L$, every data set $S \in \mathcal{X}_M^L$ can be represented by a binary vector $\mathbf{v}(\mathcal{S})$ of length $2^L$, where each non-zero entry in $\mathbf{v}(\mathcal{S})$ indicates that a specific sequence is contained in the set $\mathcal{S}$. 

The set of possible data sets can therefore be represented\footnote{This representation has been used as a proof technique in \cite{HRT17}.} by constant-weight binary vectors of length $2^L$
\[ \mathcal{V}_M^L = \{ \mathbf{v} \in \{0,1\}^{2^L} : \mathrm{wt}(\mathbf{v}) = M \}, \]
where $\mathrm{wt}(\mathbf{v})$ denotes the \textit{Hamming weight} of $\mathbf{v}$, i.e., the number of non-zero entries inside the vector $\mathbf{v}$. That is, the mapping $\mathbf{v}$ defines an isomorphism between $\mathcal{X}_M^L$ and $\mathcal{V}_M^L$ and thus $\ve{v}^{-1}$ is well-defined. Using this representation, a loss of a sequence $\ve{x} \in \mathcal{S}$ corresponds to an asymmetric $1 \rightarrow 0$ error inside $\ve{v}(\mathcal{S})$ at the position corresponding to $\ve{x}$. Substitution errors inside a sequence $\ve{x} \in \mathcal{S}$ translate to single errors in the Johnson graph in $\ve{v}(\mathcal{S})$, i.e. a single $1 \rightarrow 0$ at the corresponding position of the original sequence $\ve{x}$, and a single $0 \rightarrow 1$ error at the position of its erroneous outcome $\ve{x}'$. In case, the erroneous outcome $\ve{x}'$ is already present in $\mathcal{S}'$, the $0 \rightarrow 1$ error is omitted and there is only a single asymmetric $1 \rightarrow 0$ error at the position of the original sequence $\ve{x}$, similar to a loss of a sequence. For codes in the Johnson graph, the reader is referred to, e.g., \cite{BSSS90}
\begin{example} \label{ex:set:mapping}
	Consider the following $M=3$ stored sequences $\mathcal{S} = \{(001), (010), (110)\}$, each of length $L=3$. We choose $\ve{v}(\mathcal{S})$ to map each sequence $\ve{x} \in \mathcal{S}$ to its decimal equivalent by standard base conversion and let $\ve{v}(\mathcal{S})$ be non-zero at exactly these indices. Hence, e.g., the sequence $(110)$ is mapped to $1\cdot 2^2 + 1 \cdot 2^1 + 0 \cdot 2^0 + 1 = 7$ and thus $\ve{v}(\mathcal{S})$ will be non-zero at index $7$. Note that we additionally add $1$, since we index vectors starting by $1$. Therefore, $\ve{v}(\mathcal{S}) = (01100010)$. Assume now, the set $\mathcal{S}$ is transmitted over a $(1,1,2)_\sub$ channel, resulting in $\mathcal{S}' = \{ (001),(111) \}$, where the sequence $(110)$ was lost and the sequence $(010)$ has been perturbed by two substitution errors. The corresponding binary representation is $\ve{v}(\mathcal{S}') = (01000001)$, where there was a single $1\rightarrow 0$ at position $7$ due to the loss of the sequence $(110)$ and $1\rightarrow 0$ and $0 \rightarrow 1$ errors at positions $3$, respectively $8$, since the sequence $(010)$ was distorted to the sequence $(111)$.
\end{example}

With this principle in mind, we define a code that can correct asymmetric errors and errors in the Johnson graph.
\begin{defn}
	For all $M,L$ and positive integers $s,t$, we define $\Code_M^L(s,t)\subseteq \mathcal{V}_M^L$ to be a code of length $2^L$ that consists of codewords with constant Hamming weight $M$, which corrects $s$ asymmetric $1 \rightarrow 0$ errors and $t$ errors in the Johnson graph.
\end{defn}
With such a code $\Code_M^L(s,t)\subseteq \mathcal{V}_M^L$ in hand that can correct asymmetric errors and errors in the Johnson graph, we can construct a code for the DNA storage channel.
\begin{construction} \label{con:constant:weight}
 For all $M,L$, we define the following code
	$$ \Code_2(M,L,s,t) = \{ \mathcal{S} \in \mathcal{X}_M^L : \ve{v}(\mathcal{S}) \in \Code_M^L(s, t) \}. $$
\end{construction}
By this construction, given a constant-weight code $\Code_M^L(s,t)$, we construct the DNA storage code $\Code_2(M,L,s,t)$ by mapping each $\ve{c} \in \Code_M^L(s,t)$ to its corresponding set $\mathcal{S} = \ve{v}^{-1}(\ve{c})$. Note that this mapping can be efficiently implemented, by, e.g., a decimal to binary mapping of the non-zero positions in $\ve{c}$, as illustrated in Example \ref{ex:set:mapping}. The correctness of the construction is established in the following lemma. 
\begin{lemma} \label{lemma:constant:weight:correcting}
For all $M,L$ and positive integers $s,t$, the code $\Code_2(M,L,s,t) $ is an $(s, t, \bullet)_\mathbb{L}$-correcting code.
\end{lemma}
\begin{proof}
	Denote by $\mathcal{S}'$ the received set after a loss of at most $s$ sequences and errors in at most $t$ sequences. Let $s'$ be the number of asymmetric errors and $t'$ be the number of errors in $\ve{v}(\mathcal{S})$ with $s'+t' \leq s+t$ and $t'\leq t$. Note that $s' = M - \mathrm{wt}(\ve{v}(\mathcal{S}'))$ is detectable by the decoder. If $s' \leq s$, then the decoder can directly decode the loss of $s'\leq s$ sequences and $t'\leq t$ errors in the Johnson graph. If $s' > s$, the decoder adds $s'-s$ (arbitrarily placed) ones to $\ve{v}(\mathcal{S}')$, resulting in a loss of exactly $s$ sequences and at most $t'+s'-s \leq t$ errors in the Johnson graph.
\end{proof}
To obtain a code based on Construction \ref{con:constant:weight}, we use the fact that an asymmetric error can be represented by a single substitution error and an error in the Johnson graph can be represented by two substitution errors. With an appropriate minimum distance, it is therefore possible to employ standard codes, which will be done in the following theorem.
\begin{thm}
There exists a construction of the code $\Code_2(M,L,s,t) $ with redundancy at most $$ r(\Code_2(M,L,s,t)) \leq (s + 2 t)L. $$
\end{thm}
\begin{proof}
	By Lemma \ref{lemma:constant:weight:correcting}, it is sufficient to find a sufficiently large $M$-constant-weight code which can correct $s+2t$ substitution errors. This is since each loss in $\mathcal{S}$ causes an $1\rightarrow 0$ asymmetric error in $\ve{v}(\mathcal{S})$ and can be represented as a single substitution error and every error in a sequence in $\mathcal{S}$ will cause at most one $1\rightarrow 0$ and one $0 \rightarrow 1$ error in $\mathbf{v}(\mathcal{S})$ and thus can be represented by two substitution errors. Next, it is known, that there exists a $\tau$-substitution-correcting binary alternant code of length $2^L$ and dimension $2^L-\tau L$, cf. \cite{Rot06}. Due to the pigeonhole principle and since the alternant code has at most $2^{\tau L}$ cosets, there is one coset of the alternant code that contains at least $\binom{2^L}{M}\big/2^{\tau L}$ words with constant weight $M$, and therefore there exists a code $\Code_2(M,L,s,t)$ of cardinality at least $\binom{2^L}{M}\big/2^{\tau L}$.
	With this alternant code, the redundancy of Construction \ref{con:constant:weight} is therefore at most
	$$ r(\Code_2(M,L,s,t)) \leq \log \binom{2^L}{M} - \log \frac{\binom{2^L}{M}}{2^{\tau L}} = \tau L.$$
	Using $\tau = s+2t$ yields the theorem.
\end{proof}
The redundancy of Construction \ref{con:constant:weight} is lower than that of Construction \ref{con:index:rs}, especially for the considered case $M = 2^{\beta L}$. However, for Construction~\ref{con:index:rs} there exist efficient encoders and decoders while this is unclear for Construction~\ref{con:constant:weight}, also since the code length of the constant-weight code is exponential in~$L$.
\subsection{An Improved Indexed-Based Construction}
Construction~\ref{con:index:rs}, which uses indexing, is beneficial for its simplicity in the encoding and decoding procedures, however its redundancy is significantly larger than the one achieved by Construction~\ref{con:constant:weight}. On the other hand, Construction~\ref{con:constant:weight} does not provide an efficient encoder and decoder due to the lack of ordering in the set $\mathcal{S}$. In this section, we present a construction which introduces ideas from both of these methods. 

The main idea of this construction is to reduce the number of bits allocated for indexing each sequence. This allows a trade-off in redundancy with respect to $L$ and $M$. To simplify notation, we assume here that $M = 2^z$ for some $z \in \mathbb{N}$.
\begin{construction} \label{con:reduce:index}
Denote by $\ve{I}_c(i) \in \Sigma_2^{c \log M}$ the $c \log M$ most significant bits of the binary representation $\ve{I}(i)$ of $i$, where $0\leq  c < 1$ and $c \log M \in \mathbb{N}_0$. Further, for $1\leq i\leq M^{c}$, let $\ve{U}_i = \{ \ve{u}_{(i-1)M^{1-c}+1}, \dots, \ve{u}_{iM^{1-c}} \}$ denote a set of distinct sequences with the same index $\ve{I}_c(i)$, which are ordered lexicographically and form a symbol over a field of size $\binom{2^LM^{-c}}{M^{1-c}}$, where $\ve{u}_j \in \Sigma_2^{L-c\log M}$. 

For $\delta \geq 0$, let $\Code_3 (M,L,c,\delta) $ be the code defined by 
	\begin{align*}
	\Code_3 (M,L,c,\delta) &= \{  \mathcal{S}  \in \mathcal{X}_M^L : \, \ve{x}_i = (\ve{I}_c(i), \ve{u}_i) , \\ &(\ve{U}_1, \dots, \ve{U}_{M^{c}}) \in \mathsf{MDS}[M^{c},M^{c}-\delta] \}. 
	\end{align*}
	To guarantee existence of the MDS code, we require $M^{c} \leq \binom{2^LM^{-c}}{M^{1-c}}$ \cite{Rot06}. For $M= 2^{\beta L}$, $c \leq 1+ \frac{\log \frac{1-\beta}{\beta}}{\beta L}$ is sufficient.
\end{construction}

Note that there are $M^{c}$ groups of sequences which use the same index and each group contains $M^{1-c}$ sequences.
\begin{lemma} \label{lemma:reduce:index}
For all $M,L,\delta$, the code $\Code_3 (M,L,c,\delta)$ is
\begin{itemize}
	\item $(s,t,\bullet)_\mathbb{L}$-correcting for all $s+2t \leq \delta$,
	\item $(s,t,\bullet)_\ins$-correcting for all $s+t \leq \delta$,
	\item $(s,t,\bullet)_\del$-correcting for all $s+t \leq \delta$,
\end{itemize}
\end{lemma}
\begin{proof}
	The proof follows the same idea as that for Lemma \ref{lemma:index:rs:cor}. We will show that the MDS codeword $\ve{U} = (\ve{U}_1, \ve{U}_2, \dots, \ve{U}_{M^{c}})$ can be recovered from $\ve{U}' = (\ve{U}'_1, \ve{U}'_2, \dots, \ve{U}'_{M^{c}})$, where $\ve{U}'_i$ collects all sequences in $\mathcal{S}'$ which have the same index $i$, i.e. $\ve{U}'_i = \{ \mathrm{suff}_{L-c\log M}(\ve{x}') : \ve{x}' \in \mathcal{S}', \mathrm{pref}_{c\log M }(\ve{x}) = \ve{I}_c(i) \}$. Given $\mathcal{S}'$, we create the received estimate word $\ve{U}'$ by declaring all positions $i$ with
	$$ |\ve{U}_i'| \neq M^{1-c}, $$
	as erasures. The remaining positions in $\ve{U}'$ are filled with the corresponding symbols $\ve{U}'_i$. We will show that the number of erasures $s'$ and the number of errors $t'$ in $\ve{U}'$ satisfy $s' + 2t' \leq \delta$ by the following consideration. First, insert all error-free sequences $\ve{x} \in \SG$ into $\ve{U}'$. Up to this point $s' \leq s+t$ and $t'=0$, since there are $s+t$ sequences missing and all inserted sequences are error-free. Therefore, the $s+t$ affected groups, which contain less than $M^{1-c}$ sequences can be detected and declared erasures. Now, each of the $t$ erroneous sequences $\ve{x}' \in \SE'$ is inserted to $\ve{U}'$ and can have an arbitrary index $i$ due to errors. If the erroneous sequence $\ve{x}'$ has an index $i$ of an index group with $|\ve{U}_i'| = M^{1-c}-1$ elements, this group cannot be detected as erroneous anymore, as it contains now exactly  $M^{1-c}$ sequences. Consequently $t'$ increases by one and $s'$ decreases by one, as the group is erroneous but is not declared as an erasure in $\ve{U}'$ anymore. If the erroneous sequence $\ve{x}'$ has an index of an index group with $|\ve{U}_i'| = M^{1-c}$, this group will contain $M^{1-c}+1$ sequences afterwards and can be detected as erroneous and thus declared as erasure. In this case the number of erasures $s'$ increases by one. In all other cases neither $s'$ nor $t'$ change. Since $t$ sequences of $\mathcal{S}$ are erroneous the sum $s'+2t'$ can increase at most by $t$ with respect to the starting point $s+t$ and thus $s'+2t' \leq s+2t \leq \delta$, which proves the error correcting capability. 
	
	For the case of only insertion ($\ins$) and only deletion ($\del$) errors, it is possible to identify the erroneous groups by checking the length of the respective sequences to be larger (respectively smaller) than $L$. If these sequences are discarded and the corresponding groups declared as erasures, there are in total at most $s+t$ erasures inside the MDS codeword, which can be corrected, if $s+t\leq \delta$.
\end{proof}
The redundancy of Construction \ref{con:reduce:index} is stated in the following theorem.
\begin{thm} \label{thm:redundancy:reduce:index}
	The redundancy of Construction \ref{con:reduce:index} is given by
	$$ r(\Code_3 (M,L,c,\delta)) = \log \binom{2^L}{M} - (M^c-\delta) \log \binom{2^LM^{-c}}{M^{1-c}}. $$
 	For fixed $0<c<1$, $\delta \in \mathbb{N}_0$ and $0<\beta<1$, the redundancy of $\Code_3(M,L,c,\delta)$ is asymptotically 
	\begin{align*}
		r&(\Code_3(M,L,c,\delta)) = \frac{(1-c)}{2}M^c\log M + \frac{\log 2\pi}{2} M^c \\ 
		&+ \delta M^{1-c} \left( L-\log M+\log \e \right) + o(M^c+M^{1-c}),
	\end{align*}
	when $M\rightarrow \infty$ with $M = 2^{\beta L}$.
\end{thm}
The proof is given in Appendix \ref{app:proof:thm:redundancy:reduce:index}. Note that the last summand in the asymptotic expression for $\Code_3 (M,L,c,\delta)$ in Theorem \ref{thm:redundancy:reduce:index} quantifies the redundancy from the MDS construction, since it is multiplied by $\delta$, the redundancy of the MDS code. The two remaining terms therefore quantify the redundancy required for indexing. This shows that, asymptotically, for $c>0.5$ the redundancy needed for indexing dominates, as the terms for indexing scale as $M^c$ and the term for the MDS construction scales as $M^{1-c}$ and for $c<0.5$ the redundancy from the MDS construction dominates the redundancy of the overall construction.
\subsection{Concatenated Constructions}
Since the input of the DNA storage channel, $\mathcal{S} \in \mathcal{X}_M^L$ is a set of $M$ sequences, each of which has length $L$, it is possible to use a concatenated coding scheme to correct both a loss of sequences and errors inside the sequences. The concatenation can be constructed by choosing a set $\mathcal{S}_\mathrm{o}$ as a codeword from an outer code $\Code_\mathrm{o} \subseteq \mathcal{X}_M^{L_\mathrm{o}}$, where $L_\mathrm{o}<L$. Then, each sequence $\ve{x}_\mathrm{o} \in \mathcal{S}_\mathrm{o}$ is encoded with some inner block-code $\Code_\mathrm{i} \subseteq \Sigma_2^{L}$ of dimension $L_\mathrm{o}$ and length $L$. This procedure is formalized in the following construction.
\begin{construction}
	For all $M,L$, $L_\mathrm{o}<L$ and positive integers $s,t$, let $\Code_\mathrm{o} \subseteq \mathcal{X}_M^{L_\mathrm{o}}$ be an outer code and $\Code_\mathrm{i} \subseteq \Sigma_2^{L}$ be a standard block-code of dimension $L_\mathrm{o}$ and length $L$. Further, $\mathrm{en}(\cdot): \Sigma_2^{L_\mathrm{o}} \mapsto \Sigma_2^L$ is an encoder of the code $\Code_\mathrm{i}$. We define the concatenated construction as
	$$ \Code_4(M,L,\Code_\mathrm{i},\Code_\mathrm{o}) = \left\{ \mathcal{S} \in \mathcal{X}_M^L: \mathcal{S} = \bigcup_{\ve{x}_\mathrm{o} \in \mathcal{S}_\mathrm{o}} \mathrm{en}(\ve{x}_\mathrm{o}), \mathcal{S}_\mathrm{o} \in \Code_\mathrm{o}  \right\}. $$
\end{construction}
As outer code $\Code_\mathrm{o}$ it is in principle possible to use any code over $\mathcal{X}_M^L$. However, using the proposed Constructions \ref{con:index:rs}, \ref{con:constant:weight}, or \ref{con:reduce:index} it is possible to enhance the inner code to additionally correct a loss of sequences. This is done as follows.
\begin{lemma}
	Let $\Code_\mathrm{o} \subseteq \mathcal{X}_M^{L_\mathrm{o}}$ be an $(s,0,0)_\ET$-correcting code and $\Code_i \subseteq \Sigma_2^{L}$ be a block-code that can correct $\epsilon$ errors of type $\ET$. Then, $\Code_4(M,L,\Code_\mathrm{i},\Code_\mathrm{o})$ is $(s,M-s,\epsilon)_\ET$-correcting.
\end{lemma}
\begin{proof}
	The proof is immediate, since the inner code can correct all errors of type $\ET$ inside the sequences. After correcting these errors, it is possible to correct the lost sequences using the outer code.
\end{proof}
Note that such concatenated constructions are highly relevant in practice, as in the case that there are some sequences, which experienced more than $\epsilon$ errors can be corrected by the outer code, since Constructions \ref{con:index:rs}, \ref{con:constant:weight}, or \ref{con:reduce:index} can correct both a loss of sequences and errors in sequences, as long as $s+2t\leq \delta$. Such a construction has been used in \cite{GHPPS15}, where a Reed-Solomon code has been used as inner code and an indexed Reed-Solomon code has been used as outer code.
\subsection{Special Constructions}
In this section, we suggest constructions that can correct errors for some special cases of errors in the DNA storage channel. These constructions are interesting, since they provide insights about the channel and can likely be generalized to more general error types. 

The following $(0,1,1)_\del$-correcting construction is based on Varshamov-Tenengolts (VT) codes \cite{VT65,Lev66} that can correct a single insertion/deletion in one of the $M$ sequences. The VT code is defined to be all sequences which have the same checksum, that is defined as follows.
\begin{defn}
	The Varshamov-Tenegolts checksum $s_L(\ve{x})$ of $\ve{x} \in \Sigma_2^L$ is defined by\vspace{-1ex}$$ s_L(\ve{x}) = \sum_{i=1}^{L} ix_i \bmod (L+1). $$
\end{defn}
Our construction now employs the idea of using a single-erasure-correcting code over the checksums of all sequences. The insertion/deletion can then be corrected by first recovering the checksum of the distorted sequence and then using this checksum to correct the insertion/deletion. Note that this idea is similar to the concept of tensor product codes \cite{Wol06}. 
\begin{construction} \label{con:single:insertion}
	For an integer $a$, with $0\leq a \leq L$, the code construction $\Code_5(M,L,a) $ is given by
	\begin{align*}
	\Code_5(M,L,a) = \bigg\{ \mathcal{S} \in \mathcal{X}_M^L :\sum_{i=1}^{M} s_L(\ve{x}_i) \equiv a \bmod (L+1) \bigg\}.
	\end{align*}
\end{construction}
Note that the code can be extended to an arbitrary alphabet size $q$ by applying non-binary VT codes \cite{Ten84}. 
\begin{lemma}
	For all $M,L,a$, the code $\Code_5(M,L,a)$ is an $(0,1,1)_{\ins\del}$-correcting code.
\end{lemma}
\begin{proof}
	Assume there has been a single insertion or deletion in the $k$-th sequence, for $1\leq k\leq M$. After the reading process, the $M-1$ error-free sequences can be identified as they have length exactly $L$. The checksum deficiency is given by
	\[ a - \sum_{i \in \SG} s_L(\ve{x}_i) \bmod (L+1) = s_L(\ve{x}_k). \]
	The error in $\ve{x}_k$ is corrected by decoding in the VT code with checksum $s_L(\ve{x}_k)$.
\end{proof}
The redundancy of Construction \ref{con:single:insertion} is established in the following theorem.
\begin{thm}
	There exists $0\leq a \leq L$ such that the redundancy of Construction \ref{con:single:insertion} is at most
	$$ r(\Code_5(M,L,a)) \leq \log (L+1). $$
\end{thm}
\begin{proof}
	The codes $r(\Code_5(M,L,a))$ form a partition over $\mathcal{X}_M^L$ for all $0\leq a \leq L$. Since, there are $L+1$ distinct values for $a$, based on the pigeonhole principle there exists $0\leq a \leq L$ such that the cardinality of the code $\Code_5(M,L,a)$ satisfies $|\Code_5(M,L,a)| \geq \binom{2^L}{M}\big/(L+1)$ and thus its redundancy is at most $\log (L+1)$.
\end{proof}
As we show in Theorem \ref{thm:deletion:asymptotic}, the redundancy of any $(0,1,1)_\del$-correcting code is at least $\log(L)+o(1)$, and thus Construction \ref{con:single:insertion} is asymptotically optimal.

Using VT codes, we propose another construction of $(0,M,1)_\mathbb{ID}$-correcting codes. That is, the code can correct a single deletion or insertion in every sequence.
\begin{construction} \label{con:multiple:insertion}
	Let $a \in \mathbb{N}_0$, with $0\leq a \leq L$. Then,
	$$\Code_6(M,L,a) \hspace{-0.25ex}=  \hspace{-0.25ex}\{ \mathcal{S} \hspace{-0.25ex} \in \hspace{-0.25ex} \mathcal{X}_M^L  \hspace{-0.25ex}: \hspace{-0.25ex} s_L(\ve{x}_i)  \hspace{-0.25ex}\equiv \hspace{-0.25ex} a \bmod (L \hspace{-0.15ex} + \hspace{-0.15ex}1),\forall \, 1 \hspace{-0.25ex}\leq  \hspace{-0.25ex}i  \hspace{-0.25ex}\leq  \hspace{-0.25ex}M \}.$$
\end{construction}
\begin{lemma}
	The code $\Code_6(M,L,a)$ is an $(0,M,1)_{\ins\del}$-correcting code.
\end{lemma}
\begin{proof}
	All erroneous sequences can be detected by checking their length. If a sequence is erroneous, it can be corrected by decoding in the VT code with checksum $a$. Note that two distinct sequences cannot have the same erroneous outcome since they are different and belong to a single-deletion-correcting code.
\end{proof}
By Construction \ref{con:multiple:insertion}, all sequences $\ve{x}_i$ have the same checksum $a$, which allows to correct a single insertion or a single deletion in each sequence. The redundancy of Construction \ref{con:multiple:insertion} is computed in the following lemma.
\begin{thm}
	For fixed $0<\beta<1$, the redundancy of the code $\Code_6(M,L,0)$ satisfies asymptotically
$$	r(\Code_6(M,L,0))  \leq   M \log(L+1) + o(M), $$
	when $M \rightarrow \infty$ with $M=2^{\beta L}$.
\end{thm}
\begin{proof}
	It is known \cite{Lev66} that the number of words satisfying $s_L(\ve{x}) = 0 \bmod (L+1)$ is at least $2^L/(L+1)$. Each codeword of $\Code_6(M,L,a)$ is a subset of a VT code with cardinality $M$. Therefore the redundancy of Construction \ref{con:multiple:insertion} is at most
	\begin{align*}
	r(\Code_6(M,L,0))  &\leq \log \binom{2^L}{M} - \log \binom{\frac{2^L}{L+1}}{M}  \\
	& \leq M \log(L+1) +\frac{M^2\log \e}{2^L/(L+1)-M}.
	\end{align*}
	For $M = 2^{\beta L}$, $0<\beta<1$ the second term is $o(M)$, which concludes the proof.
\end{proof}
Interestingly, as has been shown in Theorem \ref{thm:deletion:scaling}, the redundancy of this construction is asymptotically optimal in terms of scaling with the parameters $M$ and $L$. Note that there is a non-asymptotic expression for the redundancy in the proof. 

The next construction can be used to correct $\epsilon$ substitution errors in each sequence. 
\begin{construction} \label{con:hamming} 
	Let $\Code[L,\epsilon] \subseteq \Sigma_2^L$ denote a binary $\epsilon$-substitution-correcting code of length $L$.
	For all $M\leq |\Code[L,\epsilon]|,L$, and $\epsilon$ we define the code
	$$\Code_7(M,L,\epsilon) = \{ \mathcal{S} \in \mathcal{X}_M^L : \mathcal{S} \subseteq \Code[L,\epsilon]\}.$$
\end{construction}
\begin{lemma}
	The code $\Code_7(M,L,\epsilon)$ is an $(0,M,\epsilon)_\sub$-correcting code.
\end{lemma}
The proof is immediate, since every sequence is a codeword of a code that can correct $\epsilon$ substitutions. Using binary alternant codes, it is possible to find a lower bound on the redundancy of Construction \ref{con:hamming}.
\begin{thm}
	There exists a construction for which the code $\Code_7(M,L,\epsilon)$ with fixed $\epsilon \in \mathbb{N}_0$ and $0<\beta<1$ has an asymptotic redundancy of at most
	$$ r(\Code_7(M,L,\epsilon))  \leq   M\epsilon \lceil\log L\rceil +  o(M),$$
	when $M \rightarrow \infty$ with $M = 2^{\beta L}$.
\end{thm}
\begin{proof}
	For $\Code[L,\epsilon]$ in Construction \ref{con:hamming} we use a binary $\epsilon$-substitution-correcting alternant code of length $L$, which has redundancy at most $\epsilon \lceil \log L \rceil$, cf. \cite[Ch. 5.5]{Rot06} and thus obtain a code $\Code_7(M,L,\epsilon)$ with redundancy at most
	\begin{align*}
	r(\Code_7(M,L,\epsilon))  &\leq \log \binom{2^L}{M} - \log \binom{2^{L-\epsilon \lceil \log L\rceil}}{M} \\
	& \leq M \epsilon \lceil \log(L+1) \rceil +\frac{M^2\log \e}{2^{L-\epsilon \lceil \log L\rceil}-M}.
	\end{align*}
	For $M = 2^{\beta L}$, $0<\beta<1$ the second term is $o(M)$, which concludes the proof.
\end{proof}

Note that Theorem \ref{thm:substitution:scaling} implies that for fixed $\epsilon$ this construction is close to optimality.

\section{Conclusion}\label{sec:concl}

In this paper, we set the foundations for codes over sets for DNA storage applications. After presenting the channel model and a new family of error-correcting codes over sets, we derived several bounds and constructions. Our bounds consist of extensions of the Gilbert-Varshamov and sphere packing bounds for the studied codes in the paper. We also proposed several constructions which can be either with or without indices or a reduced version of the indices. Lastly, we derived several more special constructions for a specific set of parameters. It has been illustrated that many of the proposed constructions are close to optimal, such as for the case of substitution, respectively single insertion or deletion errors inside all of the strands. We further have proposed several constructions that can cope with combinations of a loss of sequences and errors inside the sequences. By analyzing the sphere packing bounds and comparing them to our constructions, we have found important insights about the nature of the DNA storage channel. These include the surprising fact that correcting insertions or deletions requires less redundancy than correcting substitution errors inside the sequences.

\appendices

\section{Auxiliary Lemmas} \label{app:auxiliary:lemmas}
\begin{lemma} \label{lemma:ln:1:x}
	Let $f(n), g(n): \mathbb{N} \mapsto \mathbb{R}$ be two arbitrary functions with $f(n) = o(1)$ for $n \rightarrow \infty$. Then,
	$$ g(n)\ln\left(1+f(n)\right) = g(n)f(n) + O\left(g(n)f^2(n)\right). $$
\end{lemma}
\begin{proof}
	We use the standard bound on the natural logarithm
	$$ \frac{x}{x+1} \leq  \ln (1+x) \leq x, $$
	for all $x > -1$. Since $f(n) = o(1)$, there exists $n_0 \in \mathbb{N}$, such that $|f(n)| < 1$ for all $n\geq n_0$ and therefore
	$$ g(n)\frac{f(n)}{f(n)+1} \leq g(n)\ln\left(1+f(n)\right) \leq g(n)f(n), $$
	for all $n\geq n_0$. This allows to find an upper bound to the following limit of the first order approximation
	$$ \lim\limits_{n\rightarrow \infty} \left|\frac{g(n)\ln\left(1+f(n)\right) - g(n)f(n)}{g(n) f^2(n)} \right| \leq 1, $$
	by plugging in the lower and upper bound on $g(n)\ln(1+f(n))$, which proves the statement.
\end{proof}
\begin{lemma} \label{lemma:approx:binom}
	Let $f(n), g(n): \mathbb{N} \mapsto \mathbb{N}$ be two arbitrary functions with $g(n) = o(f(n))$ and $g(n)= \omega(1)$, when $n\rightarrow\infty$. The binomial coefficient satisfies
	$$ \log \binom{f(n)}{g(n)} = g(n) \log \frac{\e f(n)}{g(n)} + o(g(n)), $$
	when $n \rightarrow \infty$.
\end{lemma}
\begin{proof}
	Note that $g(n) = o(f(n))$ and $g(n)= \omega(1)$ automatically implies $f(n) = \omega(1)$. The binomial coefficient satisfies
	\begin{align*}
	\log \binom{f(n)}{g(n)} =& \log \frac{f(n)!}{(f(n)-g(n))!g(n)!}  \\
	=& \, g(n) \log \frac{f(n)}{g(n)} - \frac12 \log g(n) \\
	&- \left(f(n)-g(n)+\frac12\right) \log \left(1-\frac{g(n)}{f(n)}\right) + \gamma,
	\end{align*}
	where $\gamma = -\log \sqrt{2\pi} +  O(\frac{1}{g(n)})$.	Here we used a refinement \cite{Rob55} of Stirling's approximation, which states that
	$$ \sqrt{2\pi n}\left(\frac{n}{\e}\right)^n \e^{\frac{1}{12n+1}} \leq n! \leq \sqrt{2\pi n}\left(\frac{n}{\e}\right)^n \e^{\frac{1}{12n}}, $$
	for any $n \in \mathbb{N}$. Using Lemma \ref{lemma:ln:1:x}, we obtain
	\begin{align*}
		-& \left(f(n)-g(n)+\frac12\right) \log \left(1-\frac{g(n)}{f(n)}\right) \\
		 &= \log \e \left(g(n) -\frac{g^2(n)}{f(n)} +\frac{g(n)}{2f(n)} \right) + O\left( \frac{g^2(n)}{f(n)} \right) \\
		 &= g(n) \log \e + o(g(n)),
	\end{align*}
	where we used that $\frac{g(n)}{f(n)} = o(1)$. Plugging this result into the expression of the binomial coefficient and using further $\log g(n) = o(g(n))$ and $\gamma = o(g(n))$ proves the lemma.
\end{proof}
\begin{lemma} \label{lemma:asymtptotic:good:sets}
	For any fixed integer $\delta \in \mathbb{N}_0$ and any integer functions $y(M)\leq M$ and $z(L)$ with $z(L) \leq 2^L/y(M)$ for large enough $M$, the following asymptotic property holds
	\begin{align*}
		\log \frac{ \binom{2^L}{M-y(M)} \binom{2^L/z(L)}{y(M)} } {\binom{2^L}{M-\delta}} \leq & -y(M) \log \frac{z(L) y(M)}{\e M } \\
		&+ O\left( \frac{My(M)}{2^L} \right) + O(L),
	\end{align*}
	when $M \rightarrow \infty$ and $M = 2^{\beta L} $ with $0 < \beta <1$.
\end{lemma}
\begin{proof}
	The lemma can be shown directly by calculating the expression for the binomial coefficient
	\begin{align*}
		& \log \frac{ \binom{2^L}{M-y(M)} \binom{2^L/z(L)}{y(M)} } {\binom{2^L}{M-\delta}} \\
		=& \log \frac{\ff{(2^L/z(L))}{y(M)}\ff{(2^L-M+\delta)}{\delta}}{\ff{(2^L-M+y(M))}{y(M)}\ff{M}{\delta}} + \log \binom{M}{y(M)} \\
		\leq & y(M) \log\frac{2^L/z(L)}{2^L-M} + \log \binom{M}{y(M)} + O(L) \\
		\overset{(a)}{\leq}& y(M) \log  \frac{\e M}{z(L)y(M)} +  O\left(\frac{My(M)}{2^L}\right) + O(L),  \\
	\end{align*}
	where, $\ff{n}{m} = n\cdot (n-1)\dots (n-m+1)$ for $n,m\in \mathbb{N}_0$ denotes the falling factorial. In inequality $(a)$, we used Lemma \ref{lemma:ln:1:x} for the approximation of the logarithm and $\binom{n}{k}\leq \left(\frac{\e n}{k}\right)^k$ as an upper bound for the binomial coefficient.
\end{proof}
\section{Proof of Theorem \ref{thm:redundancy:reduce:index}} \label{app:proof:thm:redundancy:reduce:index}
The cardinality of Construction \ref{con:reduce:index} can be computed as follows. Each group $\ve{U}_i$ consists of $M^{1-c}$ unordered, distinct sequences, which share the same index $\ve{I}_c(i)$. In total, there are $M^c-\delta$ information groups, since $\delta$ groups are redundancy symbols of the MDS codeword. Therefore, the redundancy is
$$ r(\mathcal{C}_3 (M,L,c,\delta)) = \log \binom{2^L}{M} -  \log \binom{2^LM^{-c}}{M^{1-c}}^{M^c-\delta}.$$
Applying Stirling's approximation \cite{Rob55} onto the binomial coefficients yields
\begin{align*}
	r&(\mathcal{C}_3 (M,L,c,\delta)) = \log \binom{2^L}{M} - (M^c-\delta) \log \binom{2^LM^{-c}}{M^{1-c}} \\
	=& \frac{1-c}{2}M^c \log M + \frac{M^c-1}{2} \log \left( 1-\frac{M}{2^L} \right) -\gamma_2M^c + \gamma_1 \\
	&+ \delta \bigg( M^{1-c}L -M^{1-c}\log M - \frac{1-c}{2} \log M  \\
	& \quad\quad\,\, - \left(2^LM^{-c}-M^{1-c}+\frac12\right) \log\left( 1-\frac{M}{2^L}\right) + \gamma_2\bigg),
\end{align*}
where $\gamma_1 = -\log \sqrt{2\pi} + o(1) $ and $\gamma_2 = -\log \sqrt{2\pi} + o(1)$, when $c<1$. Note that it can be verified that for $c=1$, $\gamma_2$ has a different asymptotic behavior, i.e., $\gamma_2 = -\log \e + o(1)$. Therefore, for $c=1$, the expression for $r(\mathcal{C}_3 (M,L,c,\delta))$ yields the same redundancy as in Theorem \ref{thm:red:index:rs}. Employing Lemma \ref{lemma:ln:1:x} onto the two logarithmic terms yields
\begin{align*}
r&(\mathcal{C}_3 (M,L,c,\delta)) = \frac{1-c}{2}M^c\log M +\frac{\log 2 \pi}{2} M^c \\
&+ \delta M^{1-c} \left(L-\log M + \log \e \right) + o(M^c+M^{1-c}).
\end{align*}
\qed

\ifCLASSOPTIONcaptionsoff
  \newpage
\fi

\balance

\bibliography{IEEEabrv,ref.bib}
\bibliographystyle{IEEEtran}

\end{document}